\documentclass{easychair}

\usepackage{amssymb}
\usepackage[capitalize]{cleveref}
\usepackage{totcount}
\usepackage{annotate-equations}
\usepackage{refcount}
\usepackage{mathrsfs}
\usepackage{newfile}
\usepackage{bm}
\usepackage{algorithmic}
\usepackage{graphicx}
\usepackage{wrapfig}
\usepackage{textcomp}
\usepackage{xcolor}
\usepackage[textsize=small,disable]{todonotes}

\newcommand{\centr}{\mathsf{center}}
\newcommand{\lexeq}{\leq_\mathsf{lex}}
\newcommand{\lexlt}{<_\mathsf{lex}}
\newcommand{\lexlteq}{\lexeq}

\newcommand{\step}{\mathsf{step}}

\newcommand{\nn}{\mathbb{N}}
\newcommand{\zn}{\mathbb{Z}}
\newcommand{\qn}{\mathbb{Q}}
\newcommand{\qnz}{\mathbb{Q}_{+}}

\newcommand{\ourboldfont}[1]{\bm{#1}}

\newcommand{\bx}{\ourboldfont{x}}
\newcommand{\by}{\ourboldfont{y}}
\newcommand{\bc}{\ourboldfont{c}}

\newcommand{\bu}{\ourboldfont{u}}
\newcommand{\bv}{\ourboldfont{v}}
\newcommand{\bw}{\ourboldfont{w}}
\newcommand{\act}{\xrightarrow}

\newcommand{\cconf}[1]{%
    \mathchoice{\left\langle #1 \right\rangle}%
        {\langle #1 \rangle}%
        {\langle #1 \rangle}%
        {\langle #1 \rangle}%
}

\newcommand{\NP}{\mathsf{NP}}
\newcommand{\NEXP}{\mathsf{NEXP}}
\newcommand{\PSPACE}{\mathsf{PSPACE}}

\newcommand{\extra}{\mathit{extra}}

\newcommand{\cC}{\mathcal{C}}

\newcommand{\weight}[1]{W(#1)}
\newcommand{\fl}[1]{\mathsf{flat}(#1)}

\newcommand{\myparagraph}[1]{\vspace{0.2cm}\noindent\textbf{#1.} }

\newcommand\popleasychair[2]{#2}

\newtheorem{theorem}{Theorem}[section]
\newtheorem{remark}[theorem]{Remark}
\newtheorem{property}[theorem]{Property}
\newtheorem{example}[theorem]{Example}
\newtheorem{corollary}[theorem]{Corollary}
\newtheorem{proposition}[theorem]{Proposition}
\newtheorem{lemma}[theorem]{Lemma}
\newtheorem{definition}[theorem]{Definition}

\crefname{example}{Example}{Examples}
\Crefname{example}{Example}{Examples}
\crefname{property}{Property}{Properties}
\Crefname{property}{Property}{Properties}
\crefname{appsec}{Supplementary Material}{Supplementary Material}
\Crefname{appsec}{Supplementary Material}{Supplementary Material}

\title{General Decidability Results for Systems with Continuous Counters}

\author{
A. R. Balasubramanian\inst{1}\thanks{A part of the work was done when the author was at TUM, Germany}
\and
Matthew Hague\inst{2}
\and
Rupak Majumdar\inst{1}
\and
Ramanathan S. Thinniyam\inst{3}
\and
Georg Zetzsche\inst{1}
}

\institute{
  Max Planck Institute for Software Systems (MPI-SWS),
  Kaiserslautern, Germany\\
  \email{bayikudi@mpi-sws.org, rupak@mpi-sws.org, georg@mpi-sws.org}
\and
   Royal Holloway, University of London,
   Egham, UK\\
   \email{Matthew.Hague@rhul.ac.uk}
\and
   Uppsala University,
   Uppsala, Sweden\\
   \email{ramanathan.s.thinniyam@it.uu.se}
}

\authorrunning{Balasubramanian, Hague, Majumdar, Thinniyam and Zetzsche}
\titlerunning{Decidability Results for Continuous Counters}

\begin{document}

\maketitle

\begin{abstract}
Counters that hold natural numbers are ubiquitous in modeling and verifying software systems; for example, they model dynamic
creation and use of resources in concurrent programs.
Unfortunately, such discrete counters often lead to extremely high complexity.
Continuous counters are an efficient over-approximation of discrete counters. They are obtained by relaxing the original counters to hold values over the non-negative \emph{rational} numbers.

This work shows that continuous counters are extraordinarily
well-behaved in terms of decidability. Our main result is that, despite
continuous counters being infinite-state, the language of sequences of counter
instructions that can arrive in a given target configuration, is
regular. Moreover, a finite automaton for this language can be computed effectively. This implies that a wide variety of transition systems can be equipped with
continuous counters, while maintaining decidability of reachability properties.
Examples include higher-order recursion schemes, well-structured transition systems, and decidable
extensions of discrete counter systems.

We also prove a non-elementary lower bound for the size of the resulting finite automaton.
\end{abstract}

\section{Introduction}\label{sec:introduction}

Counters are ubiquitous in modeling and verifying software systems.
They model dynamic creation and use of resources in concurrent systems. 
For example, a fundamental type of infinite-state system with counters is a \emph{vector addition system} (VAS). 
These consist of finitely many counters that can be incremented and decremented in a coordinated way.
VAS, and related models such as Petri nets, have a long record of modeling resources in concurrent systems,
including manufacturing systems and discrete control \cite{PNDEDS}, business workflows \cite{PNBP}, hardware design \cite{PNHw},
and biological pathways \cite{PNinBio}.
In concurrent programming, they can model and analyze dynamic thread creation \cite{german1992reasoning,ABQ09,DBLP:conf/icalp/BaumannMTZ20,ganty2012algorithmic},
numerical data types \cite{hagueModelCheckingRecursive2011}, network broadcasts \cite{DBLP:conf/lics/EsparzaFM99}, and population protocols \cite{DBLP:conf/concur/EsparzaGLM15}.


However, analyzing systems with discrete counters encounters two crucial challenges.
First, when the size of the counters is very large, e.g., modeling packets in a network or individuals in a large population,
individual discrete interactions are too fine-grained.
In such situations, one would like a ``fluid limit'' that captures the limit behaviors as counters go to infinity.
Second, the complexity of reachability for the simplest model, vector addition systems (VAS), which have
natural-number counters,
is already Ackermann-complete~\cite{DBLP:conf/focs/CzerwinskiO21,DBLP:conf/focs/Leroux21,DBLP:conf/lics/LerouxS19}.
As soon as we combine such counters with more expressive
infinite-state data structures, algorithmic tools become extremely rare.
Such combinations are natural in concurrency theory: for example, combining pushdown systems with counters (pushdown VAS) lets us model systems with both recursion and concurrency,
but decidability of reachability remained a longstanding open problem until a very recent breakthrough \cite{PVASS}.
Reachability becomes undecidable very quickly when we extend systems (that themselves have decidable reachability) even with very weak counters. 
For example, lossy channel systems (LCS for short; which model unreliable communication channels) and 
higher-order pushdown automata (HOPA; a standard model for programs with higher-order recursion) are
well-known to have decidable reachability, but equipping them even with
very weak kinds of counters makes reachability undecidable (\cite[Prop.~7]{zetzsche2015arxiv} and \cite[Thm.~4]{kobayashi2019inclusion} for HOPA, and \cite[p.~7]{DBLP:conf/netys/Aiswarya20} for LCS).

\emph{Continuous counters} \cite{DavidAlla} provide a fluid relaxation of counter machines. 
In a continuous counter machine, any transition that adds a vector $\bx\in\zn^d$ to the counters,
can be executed with any \emph{non-zero fraction}: we are allowed to
pick a rational $0<\alpha\le 1$ and add
$\alpha\cdot \bx$ to the counters.
As a result, the counters may thus assume values in $\qnz$, the non-negative rationals.
This means that the possible behaviors of the relaxation are an overapproximation of the original machine.
This idea of relaxing discrete counters to continuous ones
goes back to David and Alla, who introduced continuous semantics in the context of 
studying fluid limits of Petri nets already in the 1980s~\cite{david1987continuous}.

Just as linear programming is a relaxation of integer linear programming with better algorithmic
complexity, continuous relaxations of counter machines have much better
complexity in many cases.
For example, reachability in Petri nets with continuous semantics can be decided in
polynomial time~\cite{DBLP:journals/fuin/FracaH15}; reachability in finite systems with continuous counters  
is $\NP$-complete \cite{blondinLogicsContinuousReachability2017};
reachability in pushdown continuous counter systems is $\NEXP$-complete~\cite{DBLP:journals/pacmpl/BalasubramanianMTZ24},
and there is an almost complete description
of the complexity of systems with affine continuous counters~\cite{DBLP:conf/lics/Balasubramanian24}.

\todo{Bala: Start of new material for applications}

The continuous relaxation has found applications in many fields, such as modeling client-server systems~\cite{mahulea2006performance}, chemical reaction networks~\cite{jordon2025} and biological networks~\cite{HK26}. Also, continuous Petri-nets are well-studied in biochemical settings, with a variety of semantics and interpretations~\cite{HGD08,HH18}.

In certain cases, continuous counters not only overapproximate the original system, but even allow an \textit{exact representation of system behavior}. For example,
we consider leaderless rendez-vous protocols, which are a model of distributed systems composed of arbitrarily many identical finite-state agents interacting in pairs. 
These interactions can be modeled using a (discrete) counter system where each of the counters is used to keep track of the number of agents in each state.
It turns out that important questions such as reachability and coverability for such protocols can be solved \textit{exactly} by using the continuous relaxation~\cite{DBLP:conf/fossacs/Balasubramanian21}. 
Furthermore, even more involved analysis problems such as the cut-off problem can be solved by combining the continuous relaxation along with the so-called integer relaxation~\cite{DBLP:conf/fossacs/Balasubramanian21}.
All of these algorithms run in polynomial time, which highlights the importance of the continuous relaxation as an important tool in algorithmic verification. 

Finally, continuous relaxations have proven to be useful for the problem of coverability in Petri nets. In~\cite{DBLP:conf/cav/EsparzaLMMN14}, continuous relaxations were used in an SMT approach to solve (i.e. prove safety of) 94 of 115 instances of Petri net benchmarks. 
Maintaining integer counters only helped prove 2 additional instances (with higher runtimes). In contrast, coverability tools only scaled to about 60 instances. 
The continuous relaxation has also been used as an optimization in
coverability algorithms for VAS~\cite{blondinApproachingCoverabilityProblem2016}.
Thus, better analysis tools for continuous counters will yield better algorithms for discrete counter systems in practice.

\todo{Bala: End of new material for applications}


Given these positive results, it is natural to ask how far one can push the decidability frontier
for continuous counters: Is reachability decidable if we extend well-known
infinite-state models with continuous counters?
This is an important concern in modeling and analysis, because we 
often take fluid limits for certain components while maintaining the discrete structure for other components.

Our main results show that indeed, continuous counters have extraordinarily good
decidability properties.
In a nutshell, decidable models remain decidable when continuous counters are added, as long as the models are closed under taking products with finite-state systems.

\subsection{Outline of Main Results}
To make our results precise, we need some terminology.
In a continuous counter machine, a transition that adds a vector $\bx\in\zn^d$ can be scaled by any rational $0 < \alpha \leq 1$, that is $\alpha\cdot\bx$ is added.
We refer to this as \emph{fractional firing}.
A system that has access to $d$ continuous counters has transitions which are labeled with vectors from a
finite set $\Sigma\subset \zn^d$, which represent the set of counter instructions.
If a sequence $w\in\Sigma^*$ can lead from counter configuration $\bx\in\qnz^d$
to $\by\in\qnz^d$ via fractional firing\todo{fractional firing not defined yet}\todo{Bala: Haven't we mentioned how this is done in paragraph 2 of page 2?}\todo{MH: not really -- i've added a couple of sentences to this paragraph}, then we write $\bx\xrightarrow{w}\by$.

Our main result is that for any dimension $d$, any given finite $\Sigma\subset \zn^d$ and
any $\bx,\by\in\qnz^d$, the set of all
words $w\in\Sigma^*$ with $\bx\xrightarrow{w}\by$ is an \emph{effectively
regular language}.

This is very surprising: The set of configurations attainable between $\bx$ and $\by$ is infinite, and involves arbitrarily large numbers. Thus, regularity strongly contrasts with almost all kinds of infinite-state systems studied so far: In almost all known types of infinite-state systems, the set of transition sequences between two configurations is in general non-regular. This holds for any type of discrete counters that can be incremented and decremented (such as classical VAS, but also the integer, bidirected, and reversal-bounded variants of VAS, etc.)\footnote{This is because applying transitions in discrete counter systems is forward- and backward-deterministic: In a forward- and backward-deterministic $\Sigma$-labeled transition system where between two configurations $c_1$ and $c_2$, infinitely many configurations can be visited, the set of $w\in\Sigma^*$ leading from $c_1$ to $c_2$ must be non-regular (a simple consequence of the Myhill-Nerode theorem). Similarly arguments apply to almost any type of infinite-state system.}, but also for pushdown automata\footnote{This is not to be confused with the fact that in a pushdown automaton, the set of reachable \emph{stack contents} is regular~\cite{Pushdown}: The set of transition sequences, e.g.\ from empty stack to empty stack can be the Dyck language, which is not regular.} (higher-order or not), lossy-channel systems. 
To our knowledge, the only exception to this is timed
automata~\cite{alur1994theory} (see, e.g.\ \cite{alur2004decision} for the
finite-word setting)\footnote{However, note that a substantial difference between timed automata and CVAS is that once a clock in a timed automaton surpasses all constants occurring in guards, its value becomes irrelevant. In a CVAS, when we reach any value $x$ we need at least $x$ decrement steps to come back to zero. This makes regularity in CVAS rather unexpected.}. 

Our construction yields an \emph{Ackermannian upper bound} on the size of an
NFA recognising the language. We also prove a \emph{non-elementary lower bound}
on the size of these NFAs. (This already indicates that our regularity proof
requires substantially different insights than the exponential construction for
timed automata~\cite{alur1994theory,alur2004decision}; and in fact, the proof
uses an entirely different approach.)

\subsection{Implications of the Main Result}
Our main result implies a very general decidability result for infinite-state
systems extended with continuous counters: For any class of infinite-state
systems that is closed under taking products with finite-state systems,
reachability is decidable if and only if it is decidable when equipped with
additional continuous counters.
Thus, it follows that, e.g., lossy channel systems and higher-order pushdown automata (even with
the collapse operation~\cite{DBLP:journals/tocl/HagueMOS17}), when extended with continuous counters, still have decidable
reachability.

Our result also applies to the \popleasychair{very}{} general class of well-structured transition
systems (WSTS)~\cite{DBLP:journals/iandc/AbdullaCJT00,finkel2001well}.
Here, the state space is ordered by a well-quasi
ordering (WQO) $\le$ that enjoys some compatibility with the transition
relation.  Under mild assumptions, many questions are decidable for
WSTS~\cite{DBLP:journals/iandc/AbdullaCJT00,finkel2001well}.
An example is the \emph{coverability problem}: Given a
configuration $\bc$ of a WSTS, can we reach a configuration $\bc'$ with
$\bc\le\bc'$? In this case, $\bc$ is said to be \emph{coverable}. Our result implies that
given such a WSTS whose transitions carry labels from some finite $\Sigma\subset\zn^d$,
one can even decide whether $\bc$ is coverable by a run that, working on $d$
continuous counters, reaches a given $\by\in\qnz^d$.
This is because, by our
result, the $d$ continuous counters can be simulated by finitely many states,
and WSTS are well-known to be closed under taking products with finite-state
systems~\cite{finkel2001well}.

Our general decidability result starkly contrasts with discrete counters.
There, many infinite-state models (e.g.\ LCS or HOPA) have undecidable reachability when extended with
weak kinds of counters (see references above).
Until very recently, the only known situation where
combining (discrete) counters with another type of data structure retains decidability is
pushdown systems with reversal-bounded counters~\cite{hagueModelCheckingRecursive2011}.
The recent breakthrough on PVASS reachability generalizes the result to all counters \cite{PVASS}.

Our result also makes the vast collection of algorithmic tools for
regular languages
(see, e.g.\ \cite{finautHandbook,autratHandbook,regexpHandbook,varietiesHandbook,profiniteHandbook}) available to the
languages of continuous counters.
This includes not only
combinatorial methods working with finite automata~\cite{finautHandbook} and regular expressions~\cite{autratHandbook,regexpHandbook}, but also algebraic tools working with finite
semigroups~\cite{varietiesHandbook}, and topological
approaches using topological semigroups~\cite{profiniteHandbook}.

\section{Preliminaries and Main Results}\label{sec:main-results}
In this section, we recall the basic definitions of our central object of study, namely, continuous vector addition system (CVAS for short). Then we formally state our main results.

\subsection{Preliminaries}

Throughout the paper, we use $\nn$, $\zn$, $\qn$, and $\qnz$ to denote the natural numbers, integers, rational numbers, and non-negative rational numbers, respectively.

A $d$-dimensional \emph{Continuous Vector Addition System} ($d$-CVAS or simply CVAS) is a finite set $\Sigma\subset \zn^d$ of vectors called \emph{transitions}. The operational semantics of a CVAS is given by means of its \emph{configurations}, which we define below.

A \emph{configuration} of $\Sigma$ is a tuple $\bx \in \qnz^d$, which intuitively denotes the current value of each of the $d$ continuous counters. 
Let $t \in \Sigma$ be a transition and let $\alpha \in (0,1]$ be a rational number.
A \emph{step} from a configuration $\bx$ to $\by$ by means of the pair $(\alpha, t)$ (denoted $\bx \act{\alpha t} \by)$ is possible iff $\by = \bx + \alpha t$. In this case, $\alpha$ is called the \emph{firing fraction} of this step.
Since each configuration must be in $\qnz^d$, a step is only possible if no component of the configuration becomes negative.
If $\bx\act{\alpha t}\by$ for some $\alpha\in(0,1]$, then we also write $\bx\act{t}\by$.

A \emph{run} of $\Sigma$ is a finite sequence of steps $\bx_0 \act{\alpha_1 t_1} \bx_1 \act{\alpha_2 t_2} \ldots \act{\alpha_n t_n} \bx_n$. Here, the sequence $\alpha_1 t_1, \ldots, \alpha_n t_n$ is called a \emph{firing sequence}.
If such a firing sequence exists between $\bx_0$ and $\bx_n$, we say that $\bx_n$ is \emph{reachable} from $\bx_0$ (written $\bx_0 \act{\alpha_1 t_1, \alpha_2 t_2, \dots, \alpha_n t_n} \bx_n$ or $\bx_0 \act{*} \bx_n$).

For a word $w\in\Sigma^*$ and configurations $\bx,\by\in\qnz^d$, we write $\bx\xrightarrow{w}\by$ if $w=t_1\cdots t_n$ and there exist rational fractions $\alpha_1,\ldots,\alpha_n\in(0,1]$ such that $\bx \xrightarrow{\alpha_1 t_1,\ldots,\alpha_n t_n}\by$. We also write $\bx \xrightarrow{\alpha w}\by$ where $\alpha=\alpha_1,\alpha_2,\ldots,\alpha_n$ is the sequence of firing fractions. It will always be clear from the context whether $\alpha$ is a single firing fraction or a sequence of firing fractions. In the sequel, we often view $\Sigma$ as an alphabet and write \emph{letters} to mean transitions.

Finally, for a CVAS $\Sigma\subset\zn^d$ and configurations $\bx,\by\in\qnz^d$, we define its transition language as
\[ L_{\Sigma}^{\bx,\by}:=\{w\in\Sigma^* \mid \bx\act{w}\by \}\]

\begin{example}\label{ex:running-example}
	Let $\Sigma = \{a,b,c\}$ with $a = (1,0,0)$, $b = (-1,1,0)$ and $c = (0,-1,1)$. 
	Let $\bx = \cconf{0,0,0}$ and $\by = \cconf{0,1/4,1/4}$. It is easy to see that $abbc \in L_\Sigma^{\bx,\by}$ because of the following run:
	\begin{equation*}
		\cconf{0,0,0} \act{\frac{1}{2}a} \cconf{\frac{1}{2},0,0} 
		\act{\frac{1}{4}b} 
		\cconf{\frac{1}{4},\frac{1}{4},0} 
		\act{\frac{1}{4}b} 
		\cconf{0,\frac{1}{2},0} 
		\act{\frac{1}{4}c} 
		\cconf{0,\frac{1}{4},\frac{1}{4}}
	\end{equation*}
	
	On the other hand, note that $bbc \notin L_\Sigma^{\bx,\by}$. Indeed, suppose there exist fractions $\alpha_1 > 0, \alpha_2 > 0, \alpha_3 > 0$ such that $\bx \act{\alpha_1 b, \ \alpha_2 b,\  \alpha_3 c} \by$. Note that the first counter of $\bx$ is zero and also that $\alpha_1 b$ decrements the first counter by $\alpha_1$. Hence, the first counter becomes negative after firing $\alpha_1 b$ from $\bx$, which is a contradiction to the requirement that configurations contain only non-negative values.	

\end{example}

Now that we have all the necessary definitions, we now move on to stating our main result.

\subsection{Main Result - Effective Regularity}
The main result of this paper is the following
\begin{theorem}
\label{th:main}
For every CVAS $\Sigma\subset\zn^d$ and configurations $\bx,\by\in\qnz^d$, the
language $L_{\Sigma}^{\bx,\by}$ is regular. Moreover, given $\Sigma, \bx$ and $\by$, one can effectively compute an NFA for it.
\end{theorem}

We now discuss some implications of this result. Theorem~\ref{th:main} implies
that for any class of labeled transition systems (LTS) that (i)~has decidable
reachability and (ii)~is closed under taking products with finite-state
systems, the following problem is decidable: Given an LTS $S$ over an alphabet
$\Sigma\subset\zn^d$, CVAS configurations $\bx,\by\in\qnz^d$, and states $s,t$
of $S$, is there a run of $S$ from $s$ to $t$, such that it can take the CVAS
configuration $\bx$ into $\by$?  Indeed, a decision procedure for this problem
is as follows: First, construct the NFA for $L_\Sigma^{\bx,\by}$ using
Theorem~\ref{th:main}, then take the product of this NFA with $S$ using (ii), and
then finally check reachability of this product system using (i).

Let us formalize this result in language-theoretic terms.  We say that a class
$\cC$ of languages has \emph{decidable regular intersection} if given
$L\subseteq\Sigma^*$ from $\cC$ and a regular language $R\subseteq\Sigma^*$, it
is decidable whether $L\cap R=\emptyset$. This is the case, e.g., if $\cC$ is
the class of languages of a class of LTS with the above properties (i) and
(ii).
\begin{corollary}\label{intersection-decidable}
Suppose $\cC$ is a class of languages with decidable regular intersection.
Then, given a language $L\subseteq\Sigma^*$ in $\cC$ over some
$\Sigma\subset\zn^d$, and configurations $\bx,\by\in\qnz^d$, we can decide
whether $L$ intersects $L_{\Sigma}^{\bx,\by}$.
\end{corollary}

\Cref{intersection-decidable} applies to a wide range of (languages of) infinite-state systems.
For example, it applies
to the large class of well-structured transition systems (WSTS)~\cite{DBLP:journals/iandc/AbdullaCJT00,finkel2001well}.
They can be viewed as language-accepting devices~\cite{DBLP:journals/acta/GeeraertsRB07} (by
considering all labels of runs that arrive in some upward closed set of
configurations), and since they are closed under taking products with
finite-state systems (which are well-structured), their corresponding language
class is closed under intersection with regular languages. In particular, this
means \emph{reachability is decidable in lossy channel-systems equipped with
continuous counters}.

Another application concerns higher-order pushdown automata~\cite{damm1986automata,Maslov1974,Maslov1976} (even with
the collapse operation~\cite{DBLP:journals/tocl/HagueMOS17}). Their language class is closed under regular
intersection, since their finitely many control states admit a product
construction with a finite automaton. Moreover, their emptiness problem is
decidable: With collapse, they are equivalent to higher-order recursion schemes (HORS)~\cite[Thm.\ 4.1]{DBLP:journals/tocl/HagueMOS17}, and for HORS, even monadic second-order logic is decidable~\cite[Thm.\ 16]{DBLP:conf/lics/Ong15}. Thus, \cref{intersection-decidable} implies that \emph{in higher-order
pushdown automata (even with collapse) equipped with
continuous counters, reachability is decidable}. This is in
contrast to discrete counters: It is known that already for second-order
pushdown automata (a special case of higher-order pushdown automata), even
adding very simple kinds of discrete counters leads to an undecidable
reachability problem (see \cite[Prop.~7]{zetzsche2015arxiv} and \cite[Thm.~4]{kobayashi2019inclusion}).

\subsubsection*{Continuous Petri nets and CVASS}

Closely related to CVAS, there are two other formalisms that overapproximate
discrete counter systems with continuous semantics: \emph{continuous Petri
nets} (CPN) and \emph{continuous vector addition systems with states} (CVASS).
The latter are almost the same as CVAS, except that the configurations also
feature a control-state, and each transition can only be taken in a particular
control-state (and updates the control-state). From \cref{th:main}, it follows
immediately that \emph{languages of CVASS are regular as well}, since the
language of a CVASS is just the language of a CVAS, intersected with a regular
language.

CPN are slightly different: Each transition is not just a vector in
$\zn^d$, but \emph{two} vectors $\mathsf{pre},\mathsf{post}\in\nn^d$. It is
applied by picking a fraction $\alpha\in(0,1]$, then first subtracting
$\alpha\cdot\mathsf{pre}$ and then adding $\alpha\cdot\mathsf{post}$.
Crucially, after subtracting $\alpha\cdot\mathsf{pre}$, the configuration
vector must be non-negative. While in the discrete setting, Petri nets are
easily translated into VASS, it is not obvious that CPN can be translated into
CVASS. However, Blondin and Haase~\cite[proof of
Prop.~4.1]{blondinLogicsContinuousReachability2017} provide a
language-preserving translation from CPN to CVASS. Thus, \emph{languages of
continuous Petri nets are regular as well}. 

\subsection{A Non-Elementary Lower Bound}
We now make a remark on the complexity of computing an NFA for $L_\Sigma^{\bx,\by}$. 
As we shall show in~\cref{sec:key-concepts}, the construction from Theorem~\ref{th:main}
only yields an Ackermannian upper bound for the NFA for $L_{\Sigma}^{\bx,\by}$. Usually, one would expect that if all languages of a class of systems is regular, then there should be an elementary upper bound (just as, e.g.\ in the region construction for timed automata~\cite{alur1994theory}). However, as our second result, we show that this is not the case: There is a non-elementary lower bound for the size of NFAs for $L_{\Sigma}^{\bx,\by}$.
More precisely, let $\exp_h\colon\nn\to\nn$ be the $h$-fold exponentiation function, i.e.\ $\exp_0(n)=n$ and $\exp_{h+1}(n)=2^{\exp_h(n)}$.
We define the \emph{size} of a configuration $\bx\in\qnz^d$ to be the largest absolute value of a numerator or denominator occurring in some component in $\bx$. Hence, a configuration $\bx\in\qnz^d$ of size $s$ can be represented using at most $O(\log(s)\cdot d)$ bits. 
We then show that
\begin{theorem}\label{thm:non-elem}
	For each $h \ge 2$, there is a $5h$-dimensional CVAS $\Sigma_h$ of size
	$O(h)$ with the following property: For each $n$, there are configurations $\bx_n,\by_n$ of size
	at most $n$, such that any NFA for $L_{\Sigma_h}^{\bx_n,\by_n}$ requires
	at least $\exp_h(n)$ states.
\end{theorem}
This theorem shows that CVAS can be extremely succinct in representing regular languages. This is perhaps surprising, given that checking whether a language of a CVAS is non-empty
can be done in polynomial-time~\cite{DBLP:journals/fuin/FracaH15}. 
Note that with these results, we leave a complexity gap between non-elementary
and Ackermannian. In particular, we leave open whether there is a
primitive-recursive construction of an NFA for $L_{\Sigma}^{\bx,\by}$.

\section{Examples and Behaviours of CVAS}\label{sec:intro-examples}
Before presenting our proofs, we give examples to show the behaviour of CVAS and motivate our definitions. In this section, let us consider the same CVAS
that we had discussed in~\cref{ex:running-example}. That is,  we consider the following 3-dimensional CVAS with transitions $\Sigma = \{a,b,c\}$ given by
\[
a = (1, 0, 0), \ b = (-1, 1, 0), \ c = (0, -1, 1) 
\]
Let us also set $\bx = \cconf{0,0,0}$ and $\by = \cconf{0,1/4,1/4}$. We will now observe some behaviours of this CVAS for $\bx$ and $\by$. These observations will lead to the concepts of path-schemes, bubbles, stars, and gatherings that are used in our proofs and that we shall introduce at the end of this section.

\subsection{Duplication}\label{subsec:duplication}

As a first step, we notice that we have the following run between $\bx$ and $\by$, which proves
that $abcabc \in L_\Sigma^{\bx,\by}$:
\begin{equation}\label{eq:one}
	\begin{array}{ccccccc}
		\cconf{0,0,0} 
		&\act{\frac{1}{4}a} 
		&\cconf{\frac{1}{4},0,0} 
		&\act{\frac{1}{8}b} 
		&\cconf{\frac{1}{8},\frac{1}{8},0} 
		&\act{\frac{1}{16}c} 
		&\cconf{\frac{1}{8},\frac{1}{16},\frac{1}{16}}\\
		&\act{\frac{1}{4}a} 
		&\cconf{\frac{3}{8},\frac{1}{16},\frac{1}{16}} 
		&\act{\frac{3}{8}b} 
		&\cconf{0,\frac{7}{16},\frac{1}{16}} 
		&\act{\frac{3}{16}c} 
		&\cconf{0,\frac{1}{4},\frac{1}{4}}\
	\end{array}	
\end{equation}

One of the important properties of CVAS runs is that it allows immediate duplication of transitions. 
For example, by halving the fraction of each transition in run~\ref{eq:one} and then inserting a copy of each transition next to its original, we get the following 
run over $aabbccaabbcc$:

\begin{equation*}
\cconf{0, 0, 0}
\act{
	\frac{1}{8}a
	\frac{1}{8}a
	\frac{1}{16}b
	\frac{1}{16}b
	\frac{1}{32}c
	\frac{1}{32}c
	\frac{1}{8}a
	\frac{1}{8}a
	\frac{3}{16}b
	\frac{3}{16}b
	\frac{3}{32}c
	\frac{3}{32}c
}
\cconf{0, \frac{1}{4}, \frac{1}{4}}
\end{equation*}

This shows that $aabbccaabbcc \in L_\Sigma^{\bx,\by}$. More importantly, this principle can now be applied
to $aabbccaabbcc$ itself (duplicating the second $a$ in $aa$, the second $b$ in $bb$ etc.). Using $\hat{a}$, $\hat{b}$, $\hat{c}$ to mark the new copies of $a$, $b$, and $c$, we can prove that $aa\hat{a}bb\hat{b}cc\hat{c}aa\hat{a}bb\hat{b}cc\hat{c} \in L_\Sigma^{\bx,\by}$.
Generalizing this argument allows us to prove that $aa^*bb^*cc^*aa^*bb^*cc^* \in L_\Sigma^{\bx,\by}$. In fact, whenever $w_1w_2\dots w_k \in L_\Sigma^{\bx,\by}$
for any $k$ words $w_1,w_2,\dots,w_k$, then $w_1w_1^*w_2w_2^*\dots w_k w_k^* \in L_\Sigma^{\bx,\by}$ as well.\footnote{
	This principle shows that the language of a CVAS is either empty or infinite. As a consequence, there are regular languages that are not CVAS languages.}

\subsection{Padding Runs}\label{subsec:padding}

Remarkably, we have more freedom to manipulate runs. To illustrate this, we take the run from~\ref{eq:one} over the word $abcabc$. 
We can pad this word by inserting the transitions $b,a,c$ in the middle (again marked by the hat symbol to help the reader) to get the word $abc\hat{b}\hat{a}\hat{c}abc$. We shall now modify the run from~\ref{eq:one} to get a run over $abc\hat{b}\hat{a}\hat{c}abc$, thereby showing that this new word is also in $L_\Sigma^{\bx,\by}$.

To get this new run, we need to redistribute the firing fractions of the transitions in run~\ref{eq:one}
to accommodate for their additional instances. We do this by reducing the firing fractions of the final instances of $a, b$ and $c$ and moving the slack to the inserted $bac$. 
In general, the goal is to adjust the fractions to maintain the same overall effect of each transition, but without any component going below $0$ due to the adjustments.
\begin{equation}\label{eq:two}
	\begin{array}{ccccccc}
		\cconf{0,0,0} 
		&\act{\frac{1}{4}a} 
		&\cconf{\frac{1}{4},0,0} 
		&\act{\frac{1}{8}b} 
		&\cconf{\frac{1}{8},\frac{1}{8},0} 
		&\act{\frac{1}{16}c} 
		&\cconf{\frac{1}{8},\frac{1}{16},\frac{1}{16}}\\
		&\act{\frac{1}{16}b}
		&\cconf{\frac{1}{16},\frac{1}{8},\frac{1}{16}}
		&\act{\frac{1}{8}a}
		&\cconf{\frac{3}{16},\frac{1}{8},\frac{1}{16}} 
		&\act{\frac{1}{16}c} 
		&\cconf{\frac{3}{16},\frac{1}{16},\frac{1}{8}}\\
		&\act{\frac{1}{8}a} 
		&\cconf{\frac{5}{16},\frac{1}{16},\frac{1}{8}} 
		&\act{\frac{5}{16}b} 
		&\cconf{0,\frac{3}{8},\frac{1}{8}} 
		&\act{\frac{1}{8}c} 
		&\cconf{0,\frac{1}{4},\frac{1}{4}}\
	\end{array}	
\end{equation}

Similarly, we can prove that $abc\hat{b}\hat{a}\hat{c}\hat{b}abc$ also belongs to $L_\Sigma^{\bx,\by}$.
This principle of redistributing the firing fractions is quite general and will allow us to conclude a broad result in~\cref{sec:perfect-path-scheme-decomposition} (namely~\cref{thm:lifting-runs}, the Lifting Runs Theorem) that allows us to show effective regularity.

\subsection{Limitations on Run Padding}

In the examples above, we could add new transitions, but we only added them after the first instance of each transition in the original run. More precisely, in \cref{subsec:duplication}, 
we went (using hat to mark additions) from $abcabc$ to $a\hat{a}b\hat{b}c\hat{c}a\hat{a}b\hat{b}c\hat{c}$.
Similarly, in~\cref{subsec:padding} we went from $abcabc$ to 
$abc\hat{b}\hat{a}\hat{c}abc$. In both cases, we added transitions only after they had already appeared.
This is not by accident, as the following argument shows, which limits the kinds of padding that we can perform on runs.

Suppose we consider the word $\hat{b}abcabc$. This adds a $b$ before its first appearance
in $abcabc$. It is easy to see that there \emph{cannot} be any firing sequence that
would enable this new word to start from $\cconf{0,0,0}$ and reach $\cconf{0,1/4,1/4}$. This
is because $b$ decrements the first counter and so it would be impossible to choose a positive non-zero fraction to fire $b$ from $\cconf{0,0,0}$ without making the first counter negative.
By the same argument, we can establish similar claims for $\hat{c}abcabc$ and $a\hat{c}bcabc$.

We can now see that a similar limitation also applies to the last appearance of transitions, i.e.,
we cannot always pad transitions after their last appearance in the original word.
For instance, consider the word $abcabc\hat{a}$ obtained from $abcabc$ by padding an $a$ at the end. Once again, we can see that there \emph{cannot}
be any firing sequence that would enable this new word to start from $\cconf{0,0,0}$ reach $\cconf{0,1/4,1/4}$. This is because $a$ increments the first counter and so it would be impossible
to choose a positive non-zero fraction to fire $a$ without making the first counter strictly bigger than zero.
By the same argument, we can also establish a similar claim for $abcab\hat{a}c$.

To sum up, while~\cref{subsec:duplication} and~\cref{subsec:padding} gave us quite some freedom to pad words and insert transitions, here we have seen some limits on padding, by looking at the first and last appearance of transitions. This naturally leads us to the notion of \emph{gatherings} which we discuss next.

\subsection{Gatherings and the Lifting Runs Theorem}\label{subsec:gatherings}

So far we have seen that it might be possible to adjust a run by adding additional transitions, as long as we add any new instance of a transition after its first appearance and before its last appearance in the run. This tells us that tracking the first and last appearances of a transition is important. Let us formalize this by means of \emph{first-} and \emph{last-appearance records}.

Let $w$ be a word. For each transition $a$ that appears in $w$, let $f_a \in [1,|w|]$ (resp. $\ell_a \in [1,|w|]$) be the index
of the first (resp. last) appearance of $a$ in $w$. The \emph{first-appearance} (resp.\ \emph{last-appearance}) \emph{record} of $w$ is the unique sequence of transitions $a_1,a_2,\ldots,a_n$ (resp. $b_1,b_2,\ldots,b_n$)
such that $f_{a_1} < f_{a_2} < \ldots < f_{a_n}$ (resp. $\ell_{b_1} < \ell_{b_2} < \ldots < \ell_{b_n})$.

First- and last-appearance will be used in \cref{sec:perfect-path-scheme-decomposition} to generalize the
reasoning in~\cref{subsec:padding} for showing that, under some broad conditions, if we have a run over a word $w$ between $\bx$ and $\by$, then we also have a run between $\bx$ and $\by$ over \emph{any word} $w'$ obtained by inserting transitions into $w$ whilst preserving its first- and last-appearance records (\cref{thm:gathering-lifting-runs}).

The first- and last-appearance records of a word are part of the crucial notion of gatherings: Informally, a gathering is a set of words that have a particular first- and last-appearance record and where the first appearances occur to the left of all last appearances. Formally, 
a \emph{gathering} is an expression of the form
$X_{a_1 \dots a_n}^{b_1 \ldots b_n}$
where $a_1,\ldots,a_n$ and $b_1,\ldots,b_n$ are letters such that $\{a_1,\ldots,a_n\}=\{b_1,\ldots,b_n\}$. Intuitively, $a_1, \dots, a_n$ is the order in which each transition should appear first, and $b_1, \ldots, b_n$ is the order in which each transition should appear last.
Additionally, the last instance of $b_1$ appears after the first appearance of $a_n$ in the gathering.
In full, a word $w$ over the alphabet $\{a_1,\ldots,a_n\}=\{b_1,\ldots,b_n\}$ is said to \emph{match} this gathering if (i)~all first appearances of letters are to the left of all last appearances of letters in $w$ and (ii)~it has the same first- and last-appearance records as the one specified by the gathering.  Note that each of $a_1, \ldots, a_n$ must appear in the word, and these characters may not necessarily comprise the whole of $\Sigma$. The term gathering comes from the fact that in such a word, ``everyone comes before everyone leaves''.

\begin{example}
	Suppose we have the gathering $X^{abc}_{abc}$. Then the words $abcabc, abcbcaabc$ match $X^{abc}_{abc}$,
	but the words $abcacb, abcabcacb$ do not: since the way the letters appear last in $abcacb$ and $abcabcacb$ is $a,c,b$.
	Moreover, due to condition (i), $ababc$ does not match because the last appearance of $a$ is before the first appearance of $c$.
	If we have the gathering $X^{ba}_{ab}$, then the word $abba$ matches $X^{ba}_{ab}$, 
	but $aba$ and $abab$ do not match $X^{ba}_{ab}$.
\end{example}

In \cref{thm:gathering-lifting-runs}, the Lifting Runs Theorem for Gatherings, we see that, by using these ideas, from a gathering that intersects $L_\Sigma^{\bx,\by}$ we can find an infinite regular language that is contained in $L_\Sigma^{\bx,\by}$. This allows us to make strides towards a regular representation of $L_\Sigma^{\bx,\by}$.

However, to construct a regular representation of $L_\Sigma^{\bx,\by}$, gatherings alone are too simple. As a next step, we synthesize the ideas presented so far to introduce \emph{path-schemes}. The Lifting Runs Theorem for Gatherings is generalised to path-schemes in \cref{thm:lifting-runs} (The Lifting Runs Theorem).

\subsection{Path-Schemes}\label{subsec:path-schemes}

In~\cref{subsec:gatherings}, we saw that it is important to track \emph{gatherings} in order to characterize the language of a CVAS. In~\cref{subsec:duplication} we also saw that it is sometimes necessary to track languages of the form $w_0 \Sigma_1^* w_1 \Sigma_2^* \dots \Sigma_n^* w_n$
for words $w_0,w_1,\dots,w_n$ and subsets $\Sigma_1,\Sigma_2,\dots,\Sigma_n \subseteq \Sigma$. \todo{Ram: I think there is a missing connection here about deleting letters from the $w_i^*$}
We refer to the $\Sigma_i^*$ components as \emph{stars}.
Let us now combine both of these ideas into the notion of a \emph{path-scheme}.

Let us fix a set of transitions $\Sigma$ and two configurations $\bx$ and $\by$.
A \emph{path-scheme} over $\Sigma$ defines a regular language of words and is given by an expression 
$\rho = u_0 X_1 u_1 \ldots X_n u_n$
where each $u_i \in \Sigma^*$ and each $X_i$ is a \emph{bubble}. A bubble is either,
\begin{itemize}
	\item
	a \emph{star} $X_A$ for a set of characters $A \subseteq \Sigma$, or
	\item
	a \emph{gathering}
	$X_{a_1 \ldots a_n}^{b_1 \ldots b_n}$
	where
		$\{a_1, \ldots, a_n\} = \{b_1, \ldots, b_n\} \subseteq \Sigma$ and $a_1, \ldots, a_n$ are pairwise distinct (and thus also $b_1, \ldots, b_n$).
\end{itemize}
The language of a star $L(X_A)$ is the set of words $A^*$.
The language of a gathering $L(X_{a_1 \ldots a_n}^{b_1 \ldots b_n})$ is the set of all words that match the gathering, i.e., all $w$ such that
$w = a_1 w_1 \ldots a_n w_n b_1 v_1 \ldots b_n v_n$
where for all $i$ we have $w_i \in \{a_1, \ldots, a_i\}^*$ and $v_i \in (\{a_1,\dots,a_n\} \setminus \{b_1, \ldots, b_i\})^*$.
For words $w \in L(X_{a_1 \ldots a_n}^{b_1 \ldots b_n})$, we define $\centr(w) = w_n$.
That is, $\centr(w)$ is the part of $w$ between the first occurrence of $a_n$ and the last occurrence of $b_1$ where all characters in the gathering can occur freely.

The language  $L(\rho)$ of a path-scheme $\rho = u_0 X_1 u_1 \ldots X_n u_n$ is the set of words of the form $u_0 w_1 \ldots w_n u_n$ where $w_i \in L(X_i)$ for all $i$.
In such a case, we say that $(w_1, \ldots, w_n)$ is a $\rho$-factor of $w$.
A path-scheme $\rho$ is said to be \emph{pre-perfect} if all of its bubbles are gatherings. $\rho$ is said to be \emph{perfect} if it is pre-perfect and in addition $L(\rho)\cap L_{\Sigma}^{\bx,\by}\ne\emptyset$.\label{definition-perfect}

We will prove that the language of a CVAS can be (effectively) represented as a finite union of perfect path-schemes. This will immediately imply that the language of a CVAS is (effectively) regular. In the next section, we introduce the main ideas and give an overview of this proof.
\todo{Bala: Removed the line that we here before stating in the next section, we will discuss KLM. Instead
added this new last line saying we will talk about the main ideas of the proof in the next section.}

%
%


\section{Key Concepts and Proof Overview}\label{sec:key-concepts}

\subsection{Proof Overview} 
Our NFA construction for \cref{th:main} proceeds as follows. It maintains a
list of path-schemes, which initially just consists of a single path-scheme
$X_\Sigma$, representing $\Sigma^*$. As we saw in~\cref{subsec:path-schemes}, a path-scheme $\rho$
gives rise to a set $L(\rho)\subseteq\Sigma^*$ of transition sequences. Here, it may still be the case that a path-scheme contains no word from $L_{\Sigma}^{\bx,\by}$. 

One step in our procedure is to decompose each path-scheme into finitely many
perfect path-schemes. By definition, a perfect path-scheme intersects
$L_{\Sigma}^{\bx,\by}$.  This is the \emph{first decomposition} procedure.
However, just finding perfect path-schemes using the first decomposition does not yield effective regularity:
Instead, we show in addition that each perfect path-scheme $\rho$ not only contains a
\emph{single run}, but in fact contains an infinite regular subset
$R_\rho\subseteq L(\rho)$ with $R_\rho\subseteq L_{\Sigma}^{\bx,\by}$. 

However, the set $R_\rho$ may not capture all of $L(\rho)$. Therefore, our
proof uses a \emph{second decomposition} to turn a perfect $\rho$ into finitely
many (not necessarily perfect) path-schemes $\sigma, \rho_1,\ldots,\rho_m$
where $\sigma$ captures $R_\rho$ and $\rho_1, \ldots, \rho_m$ capture the
remaining transition sequences of $\rho$. This step employs the Lifting Runs Theorem (\cref{thm:lifting-runs}) to identify $R_\rho$, from which $\sigma$ and $\rho_1, \ldots, \rho_m$ are computed.

Our construction then alternates between these two decomposition steps.
Termination is guaranteed by the fact that each decomposition step yields
path-schemes that are smaller---in some appropriate lexicographical
ordering---than the path-scheme that is decomposed.


Having described our high-level strategy, we now state the precise theorems
corresponding to the first and the second decomposition mentioned above.

\subsection{The Two Decomposition Steps}
We will describe here what our two decomposition steps
achieve. The resulting theorems
---\cref{path-scheme-decomposition,perfect-path-scheme-decomposition}---will be
stated here and proven in \cref{sec:path-scheme-decomposition} and
\cref{sec:perfect-path-scheme-decomposition}, respectively.

For the rest of this section, let us fix a CVAS with set of transitions $\Sigma$ and two configurations 
$\bx$ and $\by$. We now state the two main theorems that we need to prove our main result. In order to state these two theorems, we first need to set up a key definition. 

As mentioned before, to prove termination of our construction, we need to impose a
lexicographic ordering on path-schemes \todo{Bala: Added the next sentence for Reviewer A}. To this end, we define the \emph{weight} of a path-scheme $\weight{\rho}$ as a vector in $\nn^{|\Sigma|}$ such that the $i^{th}$ component of $\weight{\rho}$ denotes the number of bubbles in $\rho$
with exactly $i$ distinct elements in it. More formally, for each $1 \le i \le |\Sigma|$, we first let $\overline{i}$ be the vector that is $1$ in the $i$th component and $0$ elsewhere. Then, we inductively define weight as 
\begin{itemize}
	\item
	$\weight{X_A} = \overline{|A|}$.
	\item
	$\weight{X_{a_1 \ldots a_n}^{b_1 \ldots b_n}} = \overline{n}$.
	\item
	$\weight{u_0 X_1 \ldots X_n u_n} = \sum_i \weight{X_i}$.
\end{itemize}


\begin{example}
	Let $\Sigma = \{a,b,c\}$. The star $X_{a,b}$ has weight $(0,1,0)$ as it contains a single bubble with 2 elements in its support.
	Similarly, the gathering $X_{a,b,c}^{b,a,c}$ has weight $(0,0,1)$.
	For the path-scheme $aX_{a,b}cX_{a,b,c}^{b,a,c}$, its weight is the sum
	of the weights of its bubbles, which is $(0,1,1)$.
\end{example}

Now, given two path-schemes $\sigma$ and $\rho$, we say that $W(\sigma) \lexeq W(\rho)$ if $W(\sigma)$ is \emph{lexicographically smaller} than $W(\rho)$, i.e., either the $|\Sigma|^{th}$ component of $W(\sigma)$  is strictly smaller than the $|\Sigma|^{th}$ component of $W(\rho)$ or they are equal and the $(|\Sigma|-1)^{th}$ component of $W(\sigma)$ is strictly smaller than the $(|\Sigma|-1)^{th}$ component of $W(\rho)$ and so on. We say that $W(\sigma) \lexlt W(\rho)$ if $W(\sigma)$ is \emph{strictly lexicographically smaller} than $W(\rho)$, i.e.,
if $W(\sigma) \lexeq W(\rho)$ and $W(\sigma) \neq W(\rho)$.

We now introduce the two main theorems needed to prove our effective regularity result. Both of these results decompose path-schemes into ``simpler'' path-schemes
whilst still preserving certain properties regarding $L_\Sigma^{\bx,\by}$.
The first one states that any path-scheme $\rho$ can be decomposed into a finite number of simpler \emph{perfect} path-schemes
that together preserves the intersection $L(\rho) \cap L_\Sigma^{\bx,\by}$.

\begin{theorem}[Path-Scheme Decomposition]\label{path-scheme-decomposition}
	Given a path-scheme $\rho$, we can compute a finite set of perfect path-schemes $\rho_1, \ldots, \rho_m$ such that $L(\rho)\cap L_{\Sigma}^{\bx,\by} = (\bigcup_i L(\rho_i)) \cap L_{\Sigma}^{\bx,\by}$ and $\weight{\rho_i}\lexeq \weight{\rho}$ for each $i$.
\end{theorem}

Note that if $L(\rho) \cap L_{\Sigma}^{\bx,\by}$ is empty, then the above theorem guarantees that so is 
$(\bigcup_i L(\rho_i)) \cap L_{\Sigma}^{\bx,\by}$. However, since $\rho_1,\dots,\rho_m$ are required to be perfect path-schemes, in this case we will have $m = 0$. \todo{Bala: Added these two lines for Reviewer A.}

The second theorem states that any perfect path-scheme $\rho$ can be decomposed into a finite number of path-schemes such that one of them completely sits inside $L_\Sigma^{\bx,\by}$ and all the others
have strictly smaller weight than $\rho$.
This theorem relies on the Lifting Runs Theorem (\cref{thm:lifting-runs}) to identify the required path-schemes.
\begin{theorem}[Perfect Path-Scheme Decomposition]
	\label{perfect-path-scheme-decomposition}
	Given a perfect path-scheme $\rho$,
	we can compute a path-scheme $\sigma$ and a finite set of path-schemes $\rho_1, \ldots, \rho_m$ such that $L(\rho) = L(\sigma) \cup \bigcup_i L(\rho_i)$ and $L(\sigma) \subseteq L^{\bx, \by}_\Sigma$ and $\weight{\rho_i} \lexlt \weight{\rho}$ for all $i$.
\end{theorem}
We shall now see how these two theorems can be used to prove our effective regularity result.

\subsection{Proof of Effective Regularity}
Assuming the above two theorems,
we show that $L_\Sigma^{\bx,\by}$ is effectively regular. Here (and later in the proof), we employ a result by Blondin and Haase~\cite[Thm~4.9]{blondinLogicsContinuousReachability2017}, which allows us to decide, for a given CVAS $\Sigma$ and a regular language $R\subseteq\Sigma^*$, whether the language $L_{\Sigma}^{\bx,\by}$ intersects $R$.
\begin{proposition}[Blondin \& Haase]\label{regular-intersection-decidable}
Given a CVAS $\Sigma\subset\zn^d$, configurations $\bx,\by\in\qnz^d$ and a regular language $R$, the problem of deciding whether $R\cap L_{\Sigma}^{\bx,\by}\ne\emptyset$ is decidable (and $\NP$-complete).
\end{proposition}
This is because in \cite[Thm~4.9]{blondinLogicsContinuousReachability2017}, Blondin and Haase show that reachability in ``CVAS with states'', which are CVAS with finitely many control states, is $\NP$-complete.
\cref{regular-intersection-decidable} implies that given a path-scheme $\rho$, it is decidable whether $L(\rho)\cap L_{\Sigma}^{\bx,\by}\ne\emptyset$ (since $L(\rho)$ is a regular language).

To show effective regularity, using \cref{regular-intersection-decidable}, we first check if $\Sigma^* \cap L_\Sigma^{\bx,\by} \neq \emptyset$, i.e., we check if $L_\Sigma^{\bx,\by}$ is non-empty. If it is empty, then we are already done.
Hence, for the rest of this proof, we assume that $L_\Sigma^{\bx,\by} \neq \emptyset$.

We start with a tree $T_0$ and then progressively add leaves to construct trees $T_1, T_2, \dots $ with each tree $T_i$ satisfying the
following invariants:
\begin{enumerate}
	\item[C0] Each node of the tree will be labelled by one path-scheme $\rho$ such that $L(\rho) \cap L_\Sigma^{\bx,\by} \neq \emptyset$. This path-scheme will either be marked or unmarked.
	If it is at an odd level, then it will be a perfect path-scheme.
	\item[C1] For any node $\rho$ at an odd level, if its parent is $\rho'$, then
	$\weight{\rho} \lexlteq \weight{\rho'}$.
	\item[C2] For any unmarked node $\rho$ at an even level, if its parent is $\rho'$,
	then $\weight{\rho} \lexlt \weight{\rho'}$.
	\item[C3] If a node $\rho$ is marked, then $L(\rho) \subseteq L_\Sigma^{\bx,\by}$ and it will not have any children.
	\item[C4] The union of the languages at the leaf nodes will be an overapproximation of $L_\Sigma^{\bx,\by}$, i.e., a super-set of $L_\Sigma^{\bx,\by}$.
	
\end{enumerate}

We begin by setting $T_0$ to be the tree consisting of just the unmarked root node labelled with the path-scheme $X_\Sigma$. Clearly, the five conditions are satisfied by $T_0$ (where we take the root node
to be a node at level 0).
Suppose we have already constructed $T_k$. If all the leaves of $T_k$ are marked,
then we stop the construction. Otherwise, let $\rho$ be an unmarked leaf of $T_k$.
There are two possibilities:
\begin{itemize}
	\item $\rho$ is a leaf at an even level: In this case, we apply the Path-Scheme Decomposition
	theorem to $\rho$ to get a finite set of perfect path-schemes $\rho_1,\dots,\rho_m$
	such that $L(\rho)\cap L_{\Sigma}^{\bx,\by} = (\bigcup_i L(\rho_i)) \cap L_{\Sigma}^{\bx,\by}$ and $\weight{\rho_i}\lexeq \weight{\rho}$ for each $i$. We then attach each $\rho_i$ as an unmarked leaf
	to $\rho$. Let this new tree be $T_{k+1}$.
	
	By construction, the invariants C0, C1, C2 and C3 are immediately satisfied. For invariant C4, notice that $L(\rho)\cap L_{\Sigma}^{\bx,\by} = (\bigcup_i L(\rho_i)) \cap  L_{\Sigma}^{\bx,\by}$.
	Since C4 was satisfied by the tree $T_k$ and since the leaf $\rho$ in $T_k$ is replaced
	by the leaves $\rho_1,\dots,\rho_m$ in $T_{k+1}$, it follows that C4 is also satisfied by the tree $T_{k+1}$.
	
	\item $\rho$ is a leaf at an odd level: By invariant C0, $\rho$ must be a perfect path-scheme.
	We apply the Perfect Path-Scheme Decomposition theorem to $\rho$ to get a path-scheme $\sigma$ and a finite set of path-schemes $\rho_1, \ldots, \rho_m$ such that $L(\rho) = L(\sigma) \cup \bigcup_i L(\rho_i)$ and $L(\sigma) \subseteq L^{\bx, \by}_\Sigma$ and $\weight{\rho_i} \lexlt \weight{\rho}$ for all $i$. We then attach $\sigma$ as a marked leaf. Further, we add each $\rho_i$ that satisfies $L(\rho_i) \cap L_\Sigma^{\bx,\by} \neq \emptyset$ (we check this using \cref{regular-intersection-decidable}) as an unmarked leaf to $\rho$. Let this new tree be $T_{k+1}$.
	
	By construction, the invariants C0, C1, C2 and C3 are once again immediately satisfied.
	For invariant C4, note that the only reason we
	do not add some $\rho_i$ as a leaf to $\rho$ is if $L(\rho_i) \cap L_\Sigma^{\bx,\by} = \emptyset$. Hence, if $\sigma,\rho_{i_1}, \rho_{i_2},\dots, \rho_{i_q}$ are the path-schemes which
	we added as a leaf to $\rho$, then $L(\rho) \cap L_\Sigma^{\bx,\by} = (L(\sigma) \cup \bigcup_{i_j} L(\rho_{i_j})) \cap L_\Sigma^{\bx,\by}$. Since C4 was satisfied by the tree $T_k$ and since 
	the leaf $\rho$ in $T_k$ is replaced by the leaves $\sigma, \rho_{i_1}, \dots, \rho_{i_q}$ in $T_{k+1}$,
	it follows that C4 is satisfied by $T_{k+1}$ as well.
\end{itemize}

We claim that the construction of $T_0, T_1, \ldots$ halts after a finite number of steps.
Towards a contradiction, suppose $T_0, T_1, \ldots$ does not terminate.
Note that each $T_{i}$ satisfies all the five invariants, is finitely-branching and is obtained by adding a finite number of leaves to $T_{i-1}$.
Hence, from $T_0, T_1, \ldots$ we get a finitely-branching infinite tree $T$
that also satisfies all of the five invariants.
By K\"{o}nig's lemma, $T$ must have an infinite path. By invariant C3, this path has no marked nodes. By C1 and C2, the weight of an unmarked node at an even level
along this path strictly decreases (with respect to the $\lexlt$ ordering) when compared
with the weight of its grandparent. By the well-ordering of the lexicographic ordering,
this path cannot be infinite, leading to a contradiction.

Hence, the construction halts at some point (say $T_m$).
By definition of the construction, this can happen iff there are no unmarked leaves in $T_m$.
Let $\rho_1,\dots,\rho_\ell$ be the set of all leaves in $T_m$. By invariant C4, we have $L_\Sigma^{\bx,\by} \subseteq  \cup_{1 \le i \le \ell} \ L(\rho_i)$ and by invariant C3, we have
$\cup_{1 \le i \le \ell} \ L(\rho_i) \subseteq L_\Sigma^{\bx,\by}$. Hence, $L_\Sigma^{\bx,\by} = \cup_{1 \le i \le \ell} \ L(\rho_i)$,
and thus $L_\Sigma^{\bx,\by}$ is effectively regular.

We now conclude with a brief discussion on the complexity of our construction, i.e., the size of the final set of path-schemes that we obtain.

\myparagraph{Complexity} By a straightforward application of length function
theorems for ordinals~\cite[Theorem 3.5]{DBLP:books/hal/Schmitz17}, we show in
\cref{app:complexity} that the NFA for $L_{\Sigma}^{\bx,\by}$ can be computed
in Ackermannian time. Moreover, if we fix the number $h=|\Sigma|$ of
transitions in our CVAS, then we obtain a primitive recursive construction.
This is because termination in our construction is guaranteed by descent
w.r.t.\ the lexicographical ordering $(\nn^{|\Sigma|},\lexeq)$, which is the
ordinal $\omega^{|\Sigma|}$, and each construction of path-schemes runs in
elementary time. This elementary time bound in particular yields an elementary
control function for the application of the length function theorem for ordinals.  See
\cref{app:complexity} for a detailed argument.

\myparagraph{Division of Proof into Sections}
Hence, in order to complete the proof of effective regularity (\cref{th:main}), it suffices to prove Theorems~\ref{path-scheme-decomposition} and~\ref{perfect-path-scheme-decomposition}.
This is what we do now, by using Section~\ref{sec:path-scheme-decomposition} to prove Theorem~\ref{path-scheme-decomposition} and Section~\ref{sec:perfect-path-scheme-decomposition} to prove Theorem~\ref{perfect-path-scheme-decomposition}.

\section{Path-Scheme Decomposition}\label{sec:path-scheme-decomposition}

In this section, we will prove \cref{path-scheme-decomposition}. We begin by proving a slightly weaker claim for the special case of \emph{stars}.

\begin{lemma}[Star Decomposition]\label{star-decomposition}
	Given a star $X_A$, we can compute a finite set of
	pre-perfect path-schemes $\rho_1,\ldots,\rho_m$ such that $L(X_A)=L(\rho_1)\cup\cdots\cup
	L(\rho_m)$ and $\weight{\rho_i}\lexeq\weight{X_A}$.
\end{lemma}
\begin{proof}
	We proceed inductively w.r.t.\ the size of $A$. The base case is when $|A| = 0$. Hence
	$A=\emptyset$, and so we can pick the pre-perfect path-scheme $\rho = \epsilon$ consisting of
	just the empty word. It is then easy to see that $L(X_A) = A^*=\{\varepsilon\}$ and
	also that $\weight{\rho} = \weight{X_A}$.

	For the induction step, suppose $|A| > 0$. Then we first claim that we can decompose $A^*$ as
	\begin{equation}\label{eq:star-decompose}
		A^* ~~~=~~~ \bigcup_G L(G)~~\cup~~\bigcup_{\begin{smallmatrix}B\subsetneq A \\ C\subsetneq A\end{smallmatrix}} B^*C^*~~\cup~~\bigcup_{\begin{smallmatrix}a\in A,\\B\subseteq A\setminus\{a\},\\C\subseteq A\setminus\{a\}\end{smallmatrix}} B^*aC^*,
	\end{equation}
	where $G$ ranges over all gatherings $X_{a_1,\ldots,a_n}^{b_1,\ldots,b_n}$ with $A=\{a_1,\ldots,a_n\}=\{b_1,\ldots,b_n\}$.
 The inclusion ``$\supseteq$'' is trivial, so let us prove
the other direction ``$\subseteq$'', i.e., we prove that every word in $A^*$ belongs to one of the terms on the right-hand side of the above equation.

Consider a word $w\in A^*$. If some letter $a\in A$ occurs less than twice in $w$, then $w$ belongs either (i)~to $B^*aC^*$ for some $B\subseteq A\setminus\{a\}$, $C\subseteq A\setminus\{a\}$ or (ii)~to $B^*$ for some $B\subsetneq A$.
Thus, let us assume that every letter occurs at least twice in $w$. Let $A=\{a_1,\ldots,a_n\}$ and for each $1 \le i \le n$, let $f_i\in[1,|w|]$ be the position of the first occurrence of $a_i$ in $w$, and $\ell_i\in[1,|w|]$ be the last occurrence of $a_i$ in $w$. We now consider two cases:
\begin{enumerate}
\item[(i)] Some first occurrence is to the right of some last occurrence, i.e.\ there are $i,j$ such that $\ell_j<f_i$. In this case it is easy to see that $w$ belongs to $B^*C^*$, where $B=A\setminus \{a_i\}$ and $C=A\setminus\{a_j\}$ and so we are done.
\item[(ii)] Every first occurrence is to the left of every last occurrence, i.e.\ $f_i<\ell_j$ for all $i,j$. In this case we claim that there is a gathering $G$ such that $w\in L(G)$. Indeed, let $i_1,\ldots,i_n$
and $j_1,\ldots,j_n$ be permutations of $\{1,\ldots,n\}$ such that $f_{i_1} < f_{i_2} < \ldots f_{i_n} < \ell_{j_1} < \ell_{j_2} < \ldots \ell_{j_n}$. Then, it is easily seen that $w$ belongs to the gathering $X$ whose first-appearance record is $a_{i_1},\ldots,a_{i_n}$ and whose last-appearance record is
$a_{j_1},\ldots,a_{j_n}$. Hence, in this case as well, we are done.

\end{enumerate}

	With \eqref{eq:star-decompose} at our disposal, in order to prove the induction step
	for $X_A$, it suffices to prove the following: For each term $S$ that appears in the right-hand side of \eqref{eq:star-decompose}, we can decompose $S$ as a finite set of pre-perfect path-schemes $\sigma_1,\ldots,\sigma_k$ such that
	each $\weight{\sigma_i} \lexeq \weight{X_A}$. This follows by a simple case analysis on $S$. The intuition is that $S$ is either the language of a gathering $G$, or a term of the form $B^*C^*$ or $B^*aC^*$
	for some $B, C \subsetneq A$. In the first case, gatherings are pre-perfect path schemes themselves. In the latter two cases, we use the induction hypothesis to get pre-perfect path schemes for $B^*$ and $C^*$.
	From these, we can construct a set of pre-perfect path-schemes for each language of the form $B^*C^*$ and $B^*aC^*$. Note
	that inductively, the pre-perfect path-schemes for $B^*$, $C^*$ have weight
	$\lexeq \overline{|A|-1}$. Thus, a concatenation of such schemes will
	have a weight vector that is strictly less than $\overline{|A|}$. Moreover, each gathering $G$ has weight $\overline{|A|}$. Thus, every pre-perfect
	path-scheme we construct has a weight vector that is at most
	$\overline{|A|}=\weight{X_A}$. A detailed proof of this case analysis on $S$ is given in~\cref{app:star-decomposition}.

	This completes the induction step and therefore concludes the proof of the lemma.
\end{proof}

Using this lemma, we now prove~\cref{path-scheme-decomposition}.

\begin{proof}[Proof of \cref{path-scheme-decomposition}]
	First, notice that it suffices to only compute pre-perfect path-schemes
	$\rho_1,\ldots,\rho_m$ with $L(\rho)=\bigcup_{i=1}^m L(\rho_i)$ and
	$\weight{\rho_i}\lexeq\weight{\rho}$ for each $i$.
	Indeed, if we can compute such pre-perfect path-schemes, then we can check whether each
	$\rho_i$ is perfect. This amounts to
	checking whether $L(\rho_i)\cap L_{\Sigma}^{\bx,\by}\ne\emptyset$, which is decidable (even in
	$\NP$) by \cref{regular-intersection-decidable}.
	Let $\rho_{i_1},\ldots,\rho_{i_k}$ be the perfect path-schemes from this list.
	It is then easy to see that $L(\rho) \cap L_{\Sigma}^{\bx,\by} = (\cup_{i_j} L(\rho_{i_j})) \cap L_{\Sigma}^{\bx,\by}$ and $\weight{\rho_{i_j}} \lexeq \weight{\rho}$ for each $i_j$, thereby
	yielding the desired list of perfect path-schemes for $\rho$.

	Now let $\rho$ be a path-scheme. Let $X_{A_1},\ldots,X_{A_s}$ be the occurrences of
	stars in $\rho$. For each star $X_{A_i}$ in $\rho$, we compute a
	list $\sigma_{i,1},\ldots,\sigma_{i,t}$ of perfect path-schemes using \cref{star-decomposition} (by duplicating the path-schemes if necessary, we may assume that we get the same number $t$ of schemes for
	each $A_i$).

	For each function $f\colon[1,s]\to[1,t]$, consider the
	perfect path-scheme $\rho_f$ obtained from $\rho$ by replacing each
	$X_{A_i}$ with $\sigma_{i,f(i)}$. It is then easy to see that $L(\rho)=\bigcup_f
	\rho_f$, where $f$ ranges over all functions $[1,s]\to[1,t]$. It remains to argue that $\weight{\rho_f}\lexeq\weight{\rho}$.
	For this, we observe that since $\weight{\sigma_{i,f(i)}}\lexeq \weight{X_{A_i}}$, it follows that replacing $X_{A_i}$ with $\sigma_{i,f(i)}$ in any
	path-scheme will not increase the weight: This is
	because if $\bu\lexeq\bv$ for some vectors $\bu,\bv\in\nn^{k}$, then
	$\bw+\bu\lexeq\bw+\bv$ for every $\bw\in\nn^k$.
	Applying this observation $s$ times yields that $\weight{\rho_f}\lexeq\weight{\rho}$.
\end{proof}

\section{Perfect Path-Scheme Decomposition}\label{sec:perfect-path-scheme-decomposition}
In this section, we will prove Theorem~\ref{perfect-path-scheme-decomposition}.
In order to prove this theorem, we first set up some notation and introduce some key definitions.

Recall that a word $w$ is said to be less than or equal to another word $w'$ over the subword ordering (denoted by $w \preceq w'$) if $w$ can be obtained from $w'$ by deleting some letters in it.
Now, let $X$ be a gathering and let $w \in L(X)$. We define $(w)^\uparrow_X$ to be the set of all words $w' \in L(X)$ such that $\centr(w) \preceq \centr(w')$ with respect to the subword ordering $\preceq$. Hence, $(w)^\uparrow_X$ is the set
of all words in $L(X)$ whose center component is larger than or equal to the center
component of $w$.

We extend the above notion to pre-perfect path-schemes. Let $\rho = u_0 X_1 \ldots X_n u_n$ be a pre-perfect path-scheme and let $\vec{w} = (w_1, \ldots, w_n)$ be a $\rho$-factor.
We let $(\vec{w})^\uparrow_\rho$ be the set of words given by $\{u_0 w_1' u_1 w_2' \dots w_n' u_n : \forall \ 1 \le i \le n, \ w_i' \in ({w_i})^\uparrow_{X_i} \}$. Intuitively, $(\vec{w})^\uparrow_\rho$ is the set of all words obtained by replacing each $w_i$  with some $w_i' \in (w_i)^\uparrow_{X_i}$.
We call $(\vec{w})^\uparrow_\rho$ the $\rho$-upward closure of $\vec{w}$.

With these notations set up, we are now ready to state the three main ingredients which will together imply Theorem~\ref{perfect-path-scheme-decomposition}. To state these ingredients, we first set up some context.
Recall that to prove Theorem~\ref{perfect-path-scheme-decomposition}, we need to split a perfect path-scheme $\rho$
into path-schemes $\sigma,\rho_1,\dots,\rho_m$ such that $L(\sigma)$ is contained entirely in $L^{\bx,\by}_\Sigma$ and the weight of each $\rho_i$ is strictly less than the weight of $\rho$.
The three ingredients will help us find the path-schemes $\sigma,\rho_1,\dots,\rho_m$.
The first two are about finding $\sigma$ and the last one is about finding $\rho_1,\dots,\rho_m$.

\myparagraph{First Ingredient: Lifting Runs Theorem}
The Lifting Runs Theorem states that from any perfect path-scheme $\rho$, we can extract a word $w$
and a $\rho$-factor $\vec{w}$ such that the entire set $(\vec{w})^\uparrow_\rho$ is contained in $L_\Sigma^{\bx,\by}$. Intuitively this result allows us to \emph{lift} the $\rho$-factor $\vec{w}$ to its $\rho$-upward closure $(\vec{w})^\uparrow_\rho$.

\begin{theorem}[Lifting Runs Theorem] \label{thm:lifting-runs}
	Suppose $\rho$ is a perfect path-scheme. Then, we can compute a word $w \in L(\rho) \cap L^{\bx, \by}_\Sigma$ and a corresponding $\rho$-factor $\vec{w}$
	such that $(\vec{w})^\uparrow_\rho \subseteq L(\rho) \cap L^{\bx, \by}_\Sigma$.
\end{theorem}

\myparagraph{Second Ingredient: Path-Scheme for $\rho$-Upward Closures}
Given $\vec{w}$, obtained from the Lifting Runs Theorem, we can partition $L(\rho)$ into two: $(\vec{w})^\uparrow_\rho$
and $L(\rho) \setminus (\vec{w})^\uparrow_\rho$. Our second result proves that
it is always possible to construct a path-scheme $\sigma$ that captures exactly the set $(\vec{w})^\uparrow_\rho$.

\begin{theorem}[Path-Scheme Theorem for $\rho$-Upward Closures] \label{thm:path-scheme-upward-closure}
	Given a perfect path-scheme $\rho$ and a $\rho$-factor $\vec{w}$, we can compute a path-scheme $\sigma$ such that $L(\sigma) = (\vec{w})^\uparrow_\rho$.
\end{theorem}

\myparagraph{Third Ingredient: Path-Schemes for Complement of $\rho$-Upward Closures}
We characterize the complement of this $\rho$-upward closure, i.e., $L(\rho) \setminus (\vec{w})^\uparrow_\rho$, and
get the desired $\rho_1,\dots,\rho_m$ for proving Theorem~\ref{perfect-path-scheme-decomposition}.
In particular, we can compute a finite set of path-schemes
$\rho_1,\dots,\rho_m$ such that their union \emph{over-approximates} $L(\rho) \setminus (\vec{w})^\uparrow_\rho$ and the weight of each $\rho_i$ is strictly less than the weight of $\rho$.


\begin{theorem}[Path-Schemes Theorem for Complements] \label{thm:path-scheme-complement}
	Given a perfect path-scheme $\rho$ and a $\rho$-factor $\vec{w}$, we can compute a finite set of path-schemes $\rho_1, \ldots, \rho_m$ such that
	$L(\rho) \setminus (\vec{w})^\uparrow_\rho
	\subseteq \bigcup_i L(\rho_i)
	\subseteq L(\rho)$
	and $\weight{\rho_i} \lexlt \weight{\rho}$ for all $i$.
\end{theorem}

In the remainder of this section, we will prove our three ingredients.
First, we first show how they together imply~\cref{perfect-path-scheme-decomposition}. (A pictorial representation of this implication can be found in~\cref{fig:Venn-diagram}).

\begin{proof}[Proof of \cref{perfect-path-scheme-decomposition}]
Let $\rho$ be a perfect path-scheme. By the Lifting Runs Theorem,
we can compute a word $w \in L(\rho)$ and a $\rho$-factor $\vec{w}$ such that
$(\vec{w})^\uparrow_\rho \subseteq L(\rho) \cap L_\Sigma^{\bx,\by}$.
Then, by the Path-Scheme Theorem for $\rho$-Upward Closures, we can compute a path-scheme $\sigma$
such that $L(\sigma) = (\vec{w})^\uparrow_\rho$.
Finally, by the Path-Schemes Theorem for Complements, we can compute
path-schemes $\rho_1, \ldots, \rho_m$ such that $L(\rho) \setminus (\vec{w})^\uparrow_\rho \subseteq \bigcup_i L(\rho_i) \subseteq L(\rho)$
and $\weight{\rho_i} \lexlt \weight{\rho}$ for all $i$.
By construction, we have that $L(\rho) = L(\sigma) \cup \bigcup_i L(\rho_i)$, \
$L(\sigma)\subseteq L_\Sigma^{\bx,\by}$ and $\weight{\rho_i} \lexlt \weight{\rho}$ for all $i$,
proving Theorem~\ref{perfect-path-scheme-decomposition}.
\end{proof}

\usetikzlibrary{patterns}
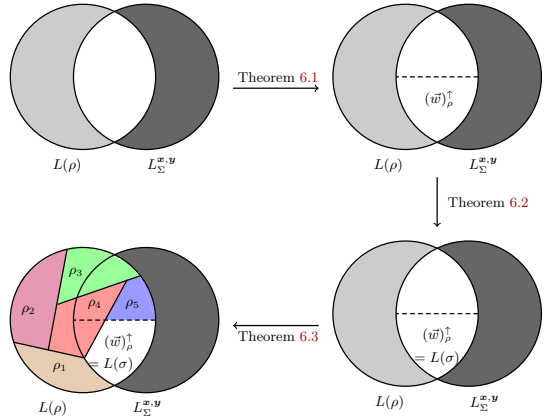
\begin{figure}[h]
	\centering
	\scalebox{0.6}{
	\begin{tikzpicture}[
		node distance=1.5cm and 2.5cm,
		box/.style={inner sep=0pt, outer sep=0pt},
		every node/.style={font=\small}
		]

		\node[box] (A) {\begin{tikzpicture}[scale = 0.7]
				\centering

		\begin{scope}
			\fill[gray!40] (-1,0) circle (2.3); 
			\fill[black!60] (1,0) circle (2.3); 
			\node at (-1.5,-2.8) {$L(\rho)$};
			\node at (1.5,-2.8) {$L^{\bx,\by}_\Sigma$};
		\end{scope}

		\begin{scope}
			\clip (-1,0) circle (2.3);
			\fill[white] (1,0) circle (2.3);
		\end{scope}

		\draw[thick] (-1,0) circle (2.3);
		\draw[thick] (1,0) circle (2.3);
	\end{tikzpicture}
};

		\node[box, right = of A] (B) {\begin{tikzpicture}[scale = 0.7]
		\centering

		\begin{scope}
			\fill[gray!40] (-1,0) circle (2.3); 
			\fill[black!60] (1,0) circle (2.3); 
			\node at (-2,-2.8) {$L(\rho)$};
			\node at (1,-2.8) {$L^{\bx,\by}_\Sigma$};
		\end{scope}


		\begin{scope}
			\clip (-1,0) circle (2.3);
			\fill[white] (1,0) circle (2.3);
			\node at (-0.4,-0.7) {$(\vec{w})^\uparrow_\rho$};
		\end{scope}


		\draw[thick,densely dashed] (1.3,0) -- (-1.3,0);
		\draw[thick] (-1,0) circle (2.3);
		\draw[thick] (1,0) circle (2.3);
	\end{tikzpicture}
};

		\node[box, below = of B] (C) {\begin{tikzpicture}[scale = 0.7]
		\centering

		\begin{scope}
			\fill[gray!40] (-1,0) circle (2.3); 
			\fill[black!60] (1,0) circle (2.3); 
			\node at (-1.5,-2.6) {$L(\rho)$};
			\node at (1.5,-2.6) {$L^{\bx,\by}_\Sigma$};
		\end{scope}


		\begin{scope}
			\clip (-1,0) circle (2.3);
			\fill[white] (1,0) circle (2.3);
			\node at (0,-0.4) {$(\vec{w})^\uparrow_\rho$};
			\node at (0,-1.2) {$= L(\sigma)$};
		\end{scope}


		\draw[thick,densely dashed] (1.3,0) -- (-1.3,0);
		\draw[thick] (-1,0) circle (2.3);
		\draw[thick] (1,0) circle (2.3);
	\end{tikzpicture}
};

		\node[box, left = of C] (D) {\begin{tikzpicture}[scale = 0.7]
		\centering

		\begin{scope}
			\fill[brown!40] (-1,0) circle (2.3); 
			\fill[black!60] (1,0) circle (2.3); 
			\node at (-1.5,-2.8) {$L(\rho)$};
			\node at (1.5,-2.8) {$L^{\bx,\by}_\Sigma$};
		\end{scope}


		\begin{scope}
			\clip (-1,0) circle (2.3);
			\fill[white] (1,0) circle (2.3);
			\node at (0.6,-0.7) {$(\vec{w})^\uparrow_\rho$};
			\node at (0.6,-1.4) {$= L(\sigma)$};
		\end{scope}


		\draw[fill = red!40, draw = none] (-2.1,-0.95) -- (-1.8,0.5) -- (0.4,1.28) -- (-0.95,-1.2);
		\draw[fill = blue!40, draw = none] (0.4,1.28) -- (0.8,1.4) -- (1.3,0) -- (-0.3,0);
		\draw[fill = blue!40, draw = none] (0.8,1.4) to [out = -45, in = 100] (1.3,0);
		\draw[fill = green!40, draw = none] (-1.8,0.5) -- (-1.5,2.25) -- (0.8, 1.4);
		\draw[fill = green!40, draw = none] (-1.5,2.25) to [out = 30, in = -10] (0.6, 1.4);
		\draw[fill = purple!40, draw = none] (-3.2,-0.7) -- (-3,1.2) -- (-1.5, 2.25) -- (-2.1,-0.95);
		\draw[fill = purple!40, draw = none] (-3.2,-0.7) -- (-3,1.2) -- (-1.5, 2.25) -- (-2.1,-0.95);
		\draw[fill = purple!40, draw = none] (-3,1.2) to [out = 60, in = 40] (-1.6, 2.15);
		\draw[fill = purple!40, draw = none] (-3.1,-0.6) to [out = -130, in = 120] (-2.9, 0.98);

		\draw[thick,densely dashed] (1.3,0) -- (-1.3,0);
		\draw[thick] (-3.2,-0.7) -- (-0.95,-1.2);
		\draw[thick] (-1.5,2.25) -- (-2.1,-0.95);
		\draw[thick] (0.8,1.4) -- (-1.8,0.5);
		\draw[thick] (0.4,1.28) -- (-0.95,-1.2);

		\begin{scope}
			\node at (0.8,0.5) {$\rho_5$};
			\node at (-0.4,0.5) {$\rho_4$};
			\node at (-1,1.5) {$\rho_3$};
			\node at (-2.5,0.3) {$\rho_2$};
			\node at (-1.5,-1.5) {$\rho_1$};
		\end{scope}

		\draw[thick] (-1,0) circle (2.3);
		\draw[thick] (1,0) circle (2.3);
	\end{tikzpicture}
};

\draw[->, thick] ($(A.east)+(0.3,0)$) -- ($(B.west)-(0.3,0)$) node[midway, above] {~\cref{thm:lifting-runs}};
\draw[->, thick] ($(B.south)-(0,0.1)$) -- ($(C.north)+(0,0.3)$) node[midway, right] {~\cref{thm:path-scheme-upward-closure}};
\draw[->, thick] ($(C.west)-(0.3,0)$) -- ($(D.east)+(0.3,0)$) node[midway, below] {~\cref{thm:path-scheme-complement}};

	\end{tikzpicture}
	}
	\caption{Venn-diagram illustration of the proof of~\cref{perfect-path-scheme-decomposition}. \cref{thm:lifting-runs} allows us to find a $\rho$-upward closed set $(\vec{w})^\uparrow_\rho$ that belongs to both the circles, i.e., to both $L(\rho)$ and $L_\Sigma^{\bx,\by}$. \cref{thm:path-scheme-upward-closure} finds a path-scheme $\sigma$ that equals this set. Along with this $\sigma$,~\cref{thm:path-scheme-complement} then allows us to split the left circle into multiple fragments (here $\rho_1,\dots,\rho_5$) such that each $\rho_i$ has weight strictly smaller than $\rho$. Notice that the fragments need not be disjoint with $\sigma$; for example, $\rho_4$ intersects with $\sigma$.}
	\label{fig:Venn-diagram}
\end{figure}

\subsection{Proof of the First Ingredient: Lifting Runs Theorem}

In this subsection, we will prove our first ingredient, namely the Lifting Runs Theorem (\cref{thm:lifting-runs}). First we prove the following special case of this theorem for \emph{gatherings}.
\begin{theorem}[Lifting Runs Theorem for Gatherings]\label{thm:gathering-lifting-runs}
	Suppose $X = X^{a_1,\ldots,a_n}_{b_1,\ldots,b_n}$ is a gathering such that $L(X) \cap L_\Sigma^{\bx,\by} \neq \emptyset$. Then, we can compute a word $w \in L(X) \cap L_\Sigma^{\bx,\by}$ such that
	$(w)^\uparrow_X \subseteq L(X) \cap L_\Sigma^{\bx,\by}$.
\end{theorem}

\begin{proof}[Proof Sketch]
	To prove this theorem, we begin with a fact from the proof of~\cite[Proposition 4.5]{blondinLogicsContinuousReachability2017}.
	It states that if $L(X) \cap L_\Sigma^{\bx,\by} \neq \emptyset$ then we can compute
	a word $w \in L(X) \cap L_\Sigma^{\bx,\by}$ such that $w = u \cdot \centr(w) \cdot v$
	where $u = a_1 a_2 \ldots a_n$, $v = b_1 b_2 \ldots b_n$, the set of letters that appear in $\centr(w)$ is exactly $\{a_1,\dots,a_n\}$ and the following property is satisfied:
	There exist two vectors $\bx', \by'$ and three sequences of firing fractions $\alpha, \beta, \gamma$ such that $\bx \xrightarrow{\alpha u }\bx' \xrightarrow{\beta \centr(w)} \by' \xrightarrow{\gamma v} \by$ and
	\begin{itemize}
		\item If a counter is non-zero at any point in the run $r_1 = \bx \xrightarrow{\alpha u} \bx'$
		then it stays non-zero from that point onwards in $r_1$.
		\item If a counter is zero at any point in the run $r_3 = \by' \xrightarrow{\gamma v} \by$
		then it stays zero from that point onwards in $r_3$.
		\item If a counter was non-zero at any point in either $r_1$ or $r_3$,
		then it stays non-zero along the run $r_2 = \bx' \xrightarrow{\beta \centr(w)} \by'$.\footnote{Technically, this fact was proven for cyclic reachability in CVAS with finitely many control states, i.e., reachability in CVAS with states with the same starting and the final state. However, the same fact applies in our setting as well, because any gathering $X$ over a CVAS $\Sigma$ can be converted into a CVAS with states where the starting and the final state are the same.}
	\end{itemize}

	With these properties at our disposal, we can prove that not only is $w \in L(X) \cap L_\Sigma^{\bx,\by}$ but that we have the stronger property $(w)^\uparrow_X \subseteq L(X) \cap L_\Sigma^{\bx,\by}$. To this end suppose $w' \in (w)^\uparrow_X$. By definition, $w' \in L(X)$
	and so it suffices to prove that $w' \in L_\Sigma^{\bx,\by}$.
	Since $X$ is a gathering and since $w' \in (w)^\uparrow_X$, we can split $w'$ as $w' = u' \cdot \centr(w') \cdot v'$ such that $\centr(w) \preceq \centr(w')$.
	Furthermore, since $u = a_1 a_2 \ldots a_n$ and $v = b_1 b_2 \ldots b_n$, by definition of a gathering,
	we also have that $u \preceq u'$ and $v \preceq v'$.

	We can now argue that there exist sequences of firing fractions $\alpha', \beta'$ and $\gamma'$
	such that $\bx \xrightarrow{\alpha' u' }\bx' \xrightarrow{\beta' \centr(w')} \by' \xrightarrow{\gamma' v'} \by$. Intuitively, we do this as follows: For each transition $t$, we show that we can take away a small part of its firing fraction from the sequences $\alpha u, \beta \centr(w)$ and $\gamma v$ and redistribute this small part amongst its additional occurrences present in $u', \centr(w')$ and $v'$ respectively. When this redistribution is done naively and arbitrarily, it might happen that some counter value becomes negative at some point.
	To prevent this, we use the special properties of $w$ to show that
	there is a way of redistribution which forces the counters to stay non-negative at each step. The formal details of this redistribution can be found in ~\cref{app:liftingruns}.
\end{proof}

\begin{example}
	We give here an example that illustrates the main techniques behind the redistribution of firing fractions in the proof of~\cref{thm:gathering-lifting-runs}.
	Let us consider the example from Section~\ref{sec:intro-examples} where $\Sigma = \{a,b,c\}$ with $a = (1,0,0), b = (-1,1,0)$ and $c = (0,-1,1)$. Let $X$ be the gathering $X^{a,b,c}_{a,b,c}$
	and let $\bx = \cconf{0,0,0}, \by = \cconf{0,\frac{1}{4},\frac{1}{4}}$. Let $w = abcbacabc \in L(X)$. Notice that $w$ can be split as $\underline{abc}\ \underline{bac}\ \underline{abc}$ where the first part is its first-appearance record, the middle part is its center and the last part is its last-appearance record.
	Corresponding to this split, we also saw the following three runs in~\cref{eq:two}
	\begin{equation}\label{eq:run-far}
		\cconf{0,0,0} \act{\frac{1}{4}a} \cconf{\frac{1}{4},0,0} \act{\frac{1}{8}b} \cconf{\frac{1}{8},\frac{1}{8},0} \act{\frac{1}{16}c} \cconf{\frac{1}{8},\frac{1}{16},\frac{1}{16}}
	\end{equation}

	\begin{equation}\label{eq:run-middle}
		\cconf{\frac{1}{8},\frac{1}{16},\frac{1}{16}} \act{\frac{1}{16}b} \cconf{\frac{1}{16},\frac{1}{8},\frac{1}{16}} \act{\frac{1}{8}a} \cconf{\frac{3}{16},\frac{1}{8},\frac{1}{16}} \act{\frac{1}{16}c} \cconf{\frac{3}{16},\frac{1}{16},\frac{1}{8}}
	\end{equation}

	\begin{equation}\label{eq:run-lar}
		\cconf{\frac{3}{16},\frac{1}{16},\frac{1}{8}} \act{\frac{1}{8}a} \cconf{\frac{5}{16},\frac{1}{16},\frac{1}{8}} \act{\frac{5}{16}b} \cconf{0,\frac{3}{8},\frac{1}{8}} \act{\frac{1}{8}c} \cconf{0,\frac{1}{4},\frac{1}{4}}
	\end{equation}

	Notice that if we set $\bx' = \cconf{\frac{1}{8},\frac{1}{16},\frac{1}{16}}$ and $\by' = \cconf{\frac{3}{16},\frac{1}{16},\frac{1}{8}}$, then this word $w$ and these three runs together satisfy the properties mentioned in the proof of~\cref{thm:gathering-lifting-runs}. Hence, according to this theorem, any word in
	$(w)^\uparrow_X$ must also be fireable from $\bx$ to reach $\by$.

	Let $w' = \underline{abac} \ \underline{bbac} \ \underline{acbc} \in (w)^\uparrow_X$. To show that ~\cref{thm:gathering-lifting-runs} applies to $w'$, we must now construct firing fractions $\alpha', \beta' $ and $\gamma'$ such that
	\[\cconf{0,0,0} \act{\alpha' abac} \cconf{\frac{1}{8},\frac{1}{16},\frac{1}{16}} \act{\beta' bbac} \cconf{\frac{3}{16},\frac{1}{16},\frac{1}{8}} \act{\gamma' acbc} \cconf{0,\frac{1}{4},\frac{1}{4}}\]

	In the proof of~\cref{thm:gathering-lifting-runs}, we get these three runs for $w'$ by modifying the corresponding runs for $w$. We first see how to construct $\alpha'$ by modifying the run in ~\cref{eq:run-far}. First, we inspect~\cref{eq:run-far} and pick a small enough $\epsilon > 0$ so that we can
	fire the modified run $(\frac{1}{4}-\epsilon) a, \frac{1}{8} b, \frac{1}{16} c$ from $\cconf{0,0,0}$ without any counter values becoming negative. It is always possible to find such an $\epsilon$ because this run satisfies the property that once a counter becomes non-zero, it stays non-zero throughout. (In this case, it suffices to take $\epsilon = \frac{1}{16}$).
	Then, we plug back the remaining $\epsilon a$ portion into this modified run
	to get the same effect as the original run. Concretely speaking we get the following run.

	\begin{equation*}
		\cconf{0,0,0} \act{\frac{3}{16}a} \cconf{\frac{3}{16},0,0} \act{\frac{1}{8}b} \cconf{\frac{1}{16},\frac{1}{8},0} \act{\frac{1}{16}a} \cconf{\frac{1}{8},\frac{1}{8},0} \act{\frac{1}{16}c} \cconf{\frac{1}{8},\frac{1}{16},\frac{1}{16}}
	\end{equation*}

	Doing exactly the same procedure, we can modify~\cref{eq:run-middle} to get the following run.

	\begin{equation*}
		\cconf{\frac{1}{8},\frac{1}{16},\frac{1}{16}} \act{\frac{1}{32}b}
		\cconf{\frac{3}{32},\frac{3}{32},\frac{1}{16}} \act{\frac{1}{32}b} \cconf{\frac{1}{16},\frac{1}{8},\frac{1}{16}} \act{\frac{1}{8}a} \cconf{\frac{3}{16},\frac{1}{8},\frac{1}{16}} \act{\frac{1}{16}c} \cconf{\frac{3}{16},\frac{1}{16},\frac{1}{8}}
	\end{equation*}

	Finally, we do the same procedure, but in reverse for the last run. More precisely, we pick $\epsilon > 0$ so that it is possible
	to fire the modified run $\frac{1}{8}a, \frac{5}{16} b, (\frac{1}{8}-\epsilon) c$ to reach $\cconf{0,\frac{1}{4},\frac{1}{4}}$ without any counter values becoming negative.  It is always possible to take such an $\epsilon > 0$ because this run satisfies the property that once a counter becomes zero, it stays zero throughout. (In this case, we see that it suffices to take $\epsilon = \frac{1}{32}$). We plug the remaining $\epsilon c$ portion into this modified run
	to get the following

	\begin{equation*}
		\cconf{\frac{3}{16},\frac{1}{16},\frac{1}{8}} \act{\frac{1}{8}a}
		 \cconf{\frac{5}{16},\frac{1}{16},\frac{1}{8}} \act{\frac{1}{32}c}
		 \cconf{\frac{5}{16},\frac{1}{32},\frac{5}{32}} \act{\frac{5}{16}b} \cconf{0,\frac{11}{32},\frac{5}{32}} \act{\frac{3}{32}c} \cconf{0,\frac{1}{4},\frac{1}{4}}
	\end{equation*}

	Altogether, the above discussion shows that we have the required three runs for $w'$. This proves that that $w' \in L^{\bx,\by}_\Sigma$.

	We stress that this procedure of redistributing runs is possible only because of the special properties of the runs. For instance, suppose instead of~\cref{eq:run-far}, we had the following run
	\begin{equation*}
		\cconf{0,0,0} \act{\frac{1}{4}a} \cconf{\frac{1}{4},0,0} \act{\frac{1}{4}b} \cconf{0,\frac{1}{4},0} \act{\frac{1}{8}c} \cconf{0,\frac{1}{8},\frac{1}{8}}
	\end{equation*}

	Then, notice that we cannot redistribute the fraction $1/4$ for $a$ without either making the first counter go below zero or without also redistributing the
	fraction $1/4$ for $b$. This means that it would not be possible to modify this run to also get a run for $abac$ between $\cconf{0,0,0}$ and $\cconf{0,\frac{1}{8},\frac{1}{8}}$ by doing the same procedure as above.
\end{example}


We can now prove the Lifting Runs theorem for any perfect path-scheme.

\begin{proof}[Proof of Theorem~\ref{thm:lifting-runs}]
	Let $\rho$ be any perfect path-scheme of the form $\rho = u_0 X_1 \dots X_n u_n$.
	Since it is perfect, there is a word $w \in L(\rho) \cap L_\Sigma^{\bx,\by}$.
	Let $\vec{w} = (w_1,\dots,w_n)$ be a $\rho$-factor of $w$.
	Hence, for each $0 \le i \le n$, there exists $\bx_i, \by_i$ with $\bx_0 = \bx$ and $\by_n = \by$ such that
	\begin{equation}\label{eq:lifting-runs-1}
		u_0 \in L_\Sigma^{\bx_0,\by_0}, \ w_1 \in L_{\Sigma}^{\by_0,\bx_1}, \ u_1 \in L_{\Sigma}^{\bx_1,\by_1} \ \cdots \cdots \ w_n \in L_\Sigma^{\by_{n-1},\bx_n}, \ u_n \in L_\Sigma^{\bx_n,\by_n}
	\end{equation}

	Applying~\cref{thm:gathering-lifting-runs} to $X_1,X_2,\dots,X_n$, we can compute words $w_1' \in L(X_1)$, $w_2' \in L(X_2)$, \dots, $w_n' \in L(X_n)$, such that
	\begin{equation}\label{eq:lifting-runs-2}
		(w_1')^\uparrow_{X_1} \subseteq L(X_1) \cap L_\Sigma^{\by_{0},\bx_1}, \
		(w_2')^\uparrow_{X_2} \subseteq L(X_2) \cap L_\Sigma^{\by_{1},\bx_2}, \ \cdots \cdots, \
		(w_n')^\uparrow_{X_n} \subseteq L(X_n) \cap L_\Sigma^{\by_{n-1},\bx_n}
	\end{equation}

	Let $w' = u_0 w_1' u_1 w_2' \ldots w_n' u_n$
	with the $\rho$-factor $\vec{w}' = (w_1',w_2',\ldots,w_n')$. By definition of $(\vec{w}')^\uparrow_\rho$ and ~\cref{eq:lifting-runs-1} and~\cref{eq:lifting-runs-2}
	it follows that
	$(\vec{w'})^\uparrow_\rho \subseteq L(\rho) \cap L_\Sigma^{\bx,\by}$.

	The above procedure gives a way to compute $w'$ when we know $\by_0,\bx_1, \by_1, \dots, \by_{n-1}, \bx_{n}$. It remains to compute $\by_0,\bx_1,\by_1,\ldots,\by_{n-1},\bx_{n}$. We can use an enumerate-and-check strategy: We enumerate $(2n)$-tuples of rational vectors and then check for each $0 \le i \le n$, whether $u_i \in L_\Sigma^{\bx_{i},\by_i}$ and $L(X_i) \cap L_{\Sigma}^{\by_{i-1},\bx_{i}} \neq \emptyset$. The former asks if a given word
	is in the intersection of a regular language and a CVAS language. The latter asks if
	the intersection of a regular language and a CVAS language is non-empty.
Since both these
	problems are decidable (even in $\NP$) by \cref{regular-intersection-decidable},
it follows that
	we can decide whether the current tuple of vectors satisfies the required conditions. Since $\rho$ is a perfect path-scheme,
	we are guaranteed that one of our guesses will succeed and so our algorithm will terminate at some point
	and compute the required sequence of vectors. This proves~\cref{thm:lifting-runs}.
\end{proof}

\subsection{Proof of the Second Ingredient: Path-Scheme Theorem for $\rho$-Upward Closures}

In this subsection, we will prove our second ingredient, namely the Path-Scheme Theorem for $\rho$-Upward Closures (\cref{thm:path-scheme-upward-closure}). To prove this theorem, we first need some definitions, which we state here. Let $X =  X_{a_1 \ldots a_n}^{b_1 \ldots b_n}$ be a gathering and let $A = \{a_1,\ldots, a_n\}$.  We let $\fl{X}$ be the path-scheme denoted by
\[
a_1 X_{\{a_1\}} a_2 X_{\{a_1, a_2\}} \ldots a_n
X_A
b_1 X_{A\setminus\{b_1\}}
b_2 X_{A\setminus\{b_1, b_2\}}
\ldots
X_{A\setminus\{b_1, \ldots, b_{n-1}\}}
b_n
\]
Note that $\fl{X}$ is a path-scheme that, by definition, accepts the same set of words as $X$, i.e., $L(\fl{X}) = L(X)$. Furthermore, if $w \in L(\fl{X})$ and 
$\vec{w} = (w_1,\ldots,w_{2n-1})$ is a  $\fl{X}$-factor of $w$, it is easily seen that the part $w_n$ in $\vec{w}$ corresponds to the bubble $X_A$ and is exactly equal to the center of $w$. For this reason, we will call the bubble $X_A$
as the central bubble of $\fl{X}$.

Finally, in our proof of this theorem, we will construct new path-schemes by substituting bubbles with path-schemes.
To this end, we define a notion of substitution for path-schemes as follows: Suppose $\rho = u_0 X_1 \ldots X_n u_n$ and $\rho' = v_0 Y_1 \ldots Y_m v_m$ are two path-schemes.
We define $\rho[\rho'/j]$ to be the path-scheme obtained by replacing the $j^{th}$ bubble in $\rho$ (i.e., $X_j$)
with the path-scheme $\rho'$. I.e.,
\[
\rho[\rho' / j]
= u_0 X_1 \ldots X_{j-1} u_{j-1}
v_0 Y_1 \ldots Y_m v_m
u_j X_{j+1} \ldots X_n u_n\,.
\]

Having set these definitions up, our strategy to prove~\cref{thm:path-scheme-upward-closure} is the same as the one that we used to prove~\cref{thm:lifting-runs}: First, prove the theorem for the special case of gatherings. Then use this to prove the theorem for the general case.

\begin{theorem}[Path-Scheme Theorem for $\rho$-Upward Closures of Gatherings] \label{thm:path-scheme-gathering-decomp}
	Given a gathering $X = X_{a_1 \ldots a_n}^{b_1 \ldots b_n}$ and
	a word $w \in L(X)$, we can compute a path-scheme $\rho$ such that
	$L(\sigma) = (w)^\uparrow_X$.
\end{theorem}

\begin{proof}
By the definition of $\fl{X}$ and $\xi$, it follows that $\sigma$ is precisely the path-scheme given by
\[a_1 X_{\{a_1\}} a_2 X_{\{a_1,a_2\}} \ldots a_n X_A c_1 X_A c_2 \ldots X_A c_m X_A b_1 X_{A\setminus\{b_1\}}
b_2 X_{A\setminus\{b_1, b_2\}}
\ldots
X_{A\setminus\{b_1, \ldots, b_{n-1}\}}
b_n
\]

It follows that any word $w'$ belongs to $L(\sigma)$ iff $w'$
can be written as
\[w' = a_1 w_1' a_2 w_2' \dots a_n w'' b_1 w'_{n} \dots w'_{2n-2} b_n\]
where
\begin{itemize}
	\item For each $i \in [1,n-1], w_i' \in L(X_{\{a_1,\ldots,a_i\}}) = \{a_1,\dots,a_i\}^*$,
	\item $w'' \in L(X_A c_1 X_A c_2 \ldots X_A c_m X_A) = L(\xi) \subseteq L(X_A)$,
	\item And for each $i \in [n,2n-2]$, $w_i' \in L(X_{A \setminus \{b_1,\dots,b_{i-n+1}\}}) = (A \setminus \{b_1,\dots,b_{i-n+1}\})^*$.
\end{itemize}
  
Hence, $w' \in L(\sigma)$ iff $w' \in L(X)$ and $\centr(w') \in L(\xi)$.
Since a word is in $L(\xi)$ iff it is larger than or equal to $\centr(w)$ with respect
to the subword ordering, it follows that $w' \in L(\sigma)$ iff $w' \in L(X)$ and $\centr(w) \preceq \centr(w')$. The latter is the definition of $(w)^\uparrow_X$ and so the proof is complete.
\end{proof}

We can now prove the Path-Scheme Theorem for $\rho$-Upward Closures.

\begin{proof}[Proof of~\cref{thm:path-scheme-upward-closure}]
	Let $\rho = u_0 X_1 u_1 \dots X_n u_n$ be a perfect path-scheme and take some $\rho$-factor $\vec{w} = (w_1,w_2,\dots,w_n)$. By~\cref{thm:path-scheme-gathering-decomp},
	for each $X_i$ and $w_i$, we can compute a path-scheme $\sigma_i$ such that $L(\sigma_i) = (w_i)^\uparrow_{X_i}$. By definition of $(\vec{w})^\uparrow_\rho$ it then follows that
	$(\vec{w})^\uparrow_\rho = L(\sigma)$ where $\sigma$ is given by
	$u_0 \sigma_1 u_1 \sigma_2 \dots \sigma_n u_n$.  This completes the proof of~\cref{thm:path-scheme-upward-closure}.
\end{proof}

\subsection{Proof of the Third Ingredient: Path-Scheme Theorem for Complements}

In this subsection, we prove our third ingredient, the Path-Scheme Theorem for Complements (\cref{thm:path-scheme-complement}). The strategy is similar to that of~\cref{thm:path-scheme-upward-closure}: First, prove the theorem for the special case of gatherings and then prove the general case. We begin with the case of gatherings.

\begin{theorem}[Path-Scheme Theorem for Complement over Gatherings] \label{thm:path-scheme-complement-gathering-decomp}
	Given a gathering $X = X_{a_1 \ldots a_n}^{b_1 \ldots b_n}$ and a word $w \in L(X)$,
	we can compute a finite set of path-schemes $\rho_1, \ldots, \rho_m$ such that
	$L(X) \setminus (w)^\uparrow_X = \bigcup_i L(\rho_i)$
	and $\weight{\rho_i} \lexlt \weight{X}$ for all $i$.
\end{theorem}
\begin{proof}[Proof Sketch]
	Let $\centr(w) = c_1 \ldots c_m$ and $A = \{a_1, \ldots, a_n\}$.
	For any $a \in A$, let $\overline{a} = A \setminus \{a\}$.
	Note that a word is not larger than or equal to $\centr(w)$ (with respect to the subword ordering) iff the largest prefix of $\centr(w)$ that is a subword is strictly smaller than $\centr(w)$.

	Hence, for every $i \in [1,m]$, we let $\sigma_i$ be the path-scheme given by
	$X_{\overline{c_1}} c_1 X_{\overline{c_2}} c_2 \ldots X_{\overline{c_i}}$.
	By the above, $\bigcup_{i \in [1,m]} L(\sigma_i)$ is exactly the set of all words
	that are not larger than or equal to $\centr(w)$.

	Now the set of words in $L(X) \setminus (w)^\uparrow_X$ are precisely those words in $L(X)$ whose center component
	belongs to $\bigcup_{i \in [1,m]} L(\sigma_i)$. Recall that $\fl{X}$ is a path-scheme such that $L(\fl{X}) = L(X)$. Furthermore, the central bubble
	of $\fl{X}$ captures precisely the center of all words in $L(X)$.
	Hence, if we replace the central bubble of $\fl{X}$ with $\sigma_1,\dots,\sigma_m$
	to get path-schemes $\rho_1,\dots,\rho_m$ it will follow that $L(X) \setminus (w)^\uparrow_X = \bigcup_{i \in [1,m]} L(\rho_i)$. Since the central bubble of $\fl{X}$ is its $n^{th}$ component, it then suffices to set each $\rho_i$ to be $\rho_i = \fl{X}[\sigma_i/ n]$.
	The formal proof behind the correctness of the equality $L(X) \setminus (w)^\uparrow_X = \bigcup_{i \in [1,m]} L(\rho_i)$ can be found in~\cref{app:gathering-decomp-2}.

	For any $i$, notice that any bubble of $\rho_i$ is either a bubble of $\fl{X}$ that is not the central bubble
	or a bubble of $\sigma_i$. By construction, in either case, it follows that, each bubble of $\rho_i$ is actually a star of the form $X_S$ with $|S| < |A|$. It follows then that the weight of each bubble of $\rho_i$
	has only 0 in the $j^{th}$ component for any $j \ge |A|$. Hence, the overall weight of $\rho_i$ also
	has only 0 in the $j^{th}$ component for any $j \ge |A|$. This implies that the weight of each $\rho_i$ is lexicographically strictly smaller than the weight of $X$ (which has a 1 in the $|A|^{th}$ component). This then completes the proof.
\end{proof}

We can now prove the Path-Schemes Theorem for Complements.

\begin{proof}[Proof of~\cref{thm:path-scheme-complement}]
	Let $\rho = u_0X_1u_1\dots X_n u_n$ be a perfect path-scheme and \popleasychair{}{also} let $\vec{w} = (w_1,\dots,w_n)$ be a $\rho$-factor. By~\cref{thm:path-scheme-complement-gathering-decomp},
	for each $X_i$ and $w_i$, we can compute a finite number of path-schemes $\sigma_i^1,\dots,\sigma_i^{m_i}$ such that $\bigcup_{j \in [1,m_i]} L(\sigma_i^j) = L(X_i) \setminus (w_i)^\uparrow_{X_i}$ and $\weight{\sigma_i^j} \lexlt \weight{X_i}$.
	Now, for each $i \in [1,n]$ and each $j \in [1,m_i]$, we can create
	a path-scheme $\rho_i^j = \rho[\sigma_i^j/i]$.

	First, we claim that if $w' \in L(\rho) \setminus (\vec{w})^\uparrow_\rho$,
	then $w' \in L(\rho_i^j)$ for some $i,j$. Indeed, if $w' \in L(\rho) \setminus (\vec{w})^\uparrow_\rho$
	then for any $\rho$-factor $\vec{w'} = (w'_1,\dots,w'_n)$ there must be some $i$ such that
	$\centr(w_i) \not \preceq \centr(w_i')$. Hence, $w_i' \in L(X_i) \setminus (w_i)^\uparrow_{X_i}$
	and so $w_i' \in L(\sigma_i^j)$ for some $j$. It then follows that $w' \in L(\rho_i^j)$.
	Hence, we can conclude that $L(\rho) \setminus (\vec{w})^\uparrow_\rho \subseteq \bigcup_{i,j} L(\rho_i^j)$.

	Second, since each $\rho_i^j$ was obtained from $\rho$ by substituting $X_i$ with $\sigma_i^j$
	and since $L(\sigma_i^j) \subseteq L(X_i)$, it follows that $L(\rho_i^j) \subseteq L(\rho)$.
	Hence we can conclude that $\bigcup_{i,j} L(\rho_i^j) \subseteq L(\rho)$.

	Finally for each $i,j$ we have that $\weight{\rho_i^j} = \weight{\rho} - \weight{X_i} + \weight{\sigma_i^j}$.
	Since $\weight{\sigma_i^j} \lexlt \weight{X_i}$, it follows that $\weight{\rho_i^j} \lexlt \weight{\rho}$. This completes the proof of~\cref{thm:path-scheme-complement}.
\end{proof}

\section{Comparison with the KLM Decomposition}
In the previous three sections, we have presented our effective regularity proof for CVAS languages.
Our construction bears a striking similarity with the KLM-decomposition\footnote{The authors of this work discovered the similarity only after obtaining the results of this paper. However, some structural aspects and some terms in the construction (such as ``perfect'' for path-schemes) have been chosen after this discovery to match the KLM-decomposition.} from the theory of (discrete) vector addition systems~\cite{DBLP:conf/stoc/Kosaraju82,DBLP:journals/tcs/Lambert92,DBLP:journals/siamcomp/Mayr84}.
Since we find it surprising that a KLM-like decomposition (i)~can be used for a regularity proof and (ii)~applies to the continuous setting, we use this section to point out similarities and differences between our construction and the KLM-decomposition. \emph{Readers who only want to understand the new proofs may safely skip this section.}

\myparagraph{Vector Addition Systems}
Intuitively, a vector addition system is the same as a CVAS, except that the firing fraction at each step must be exactly equal to 1.
Formally, a \emph{vector addition system} (VAS) is syntactically the same
as a CVAS, and is given by a finite set of transitions $\Sigma \subset \zn^d$. 
A configuration of a VAS is a vector in $\nn^d$. A step from 
a configuration $\bx$ to $\by$ by means of a transition $t$ is possible if and only if $\by = \bx + t$. 
Since each configuration must be an element in $\nn^d$, a step is only possible if no component of the configuration becomes negative. Similar to CVAS, we can define runs and reachability in a VAS.
The \emph{reachability problem for VAS} is to decide, given a set of transitions $\Sigma$ and two configurations $\bx$ and $\by$, whether $\bx$ can reach $\by$.

A landmark result in the theory of VAS is that reachability
is decidable and in fact,
Ackermann-complete~\cite{DBLP:conf/focs/CzerwinskiO21,DBLP:conf/focs/Leroux21,DBLP:conf/lics/LerouxS19}.
The first known algorithm for reachability in VAS is the so-called \emph{KLM-decomposition} (named after three key contributors
Kosaraju~\cite{DBLP:conf/stoc/Kosaraju82},
Lambert~\cite{DBLP:journals/tcs/Lambert92}, and
Mayr~\cite{DBLP:journals/siamcomp/Mayr84}).  

\myparagraph{Why is the Similarity Surprising?} At first glance, it might seem
natural that ideas for VAS also apply to CVAS\@. However,
compared to other VAS overapproximations like $\zn$-VASS (where the counters
can become negative), it is not obvious that a CVAS can be
translated into a VAS with the same language. For example, consider the
following decision problem: Given a VAS, configurations $\bx,\by\in\nn^d$ and a
number $\ell$ (represented in binary), decide whether there is a run of length
$\ell$ from $\bx$ to $\by$. In the VAS setting, this problem is
$\PSPACE$-complete (since all intermediate counter values are at most exponential). However, for CVAS, this problem is
$\NEXP$-complete~\cite[Theorems~1.1 and
6.1]{DBLP:journals/pacmpl/BalasubramanianMTZ24}. Hence, unless $\PSPACE=\NEXP$,
one cannot translate a CVAS into a language-equivalent VAS in polynomial time.
Of course, \cref{th:main} implies that CVAS languages are also VAS languages
(since the former are regular), but we are not aware of any proof of this fact besides
\cref{th:main}.

\myparagraph{An Overview of KLM}
Roughly speaking, the KLM-decomposition works as follows (we adopt the terminology of Leroux and
Schmitz~\cite{DBLP:conf/lics/LerouxS15}). 
It maintains a list of structures
called \emph{marked witness graph sequences} (MWGS). Each MWGS describes a set
of \emph{pre-runs}, which intuitively are sequences of transitions that are not
necessarily valid VAS runs.  In the beginning, the list consists of a single
MWGS---essentially the VAS itself.

A crucial notion for MWGSes is that of a ``perfect MWGS''. A key insight in the KLM-decomposition is that a perfect MWGS has a valid run among its pre-runs.
Moreover, whether an MWGS is perfect can be decided using simpler algorithmic
tools: (i)~a coverability check, and (ii)~feasibility of a system of integer
linear inequalities.

The KLM algorithm first checks perfectness of each of the current MWGS. If one
of them is perfect, it concludes the existence of a valid run. On the other
hand, for every MWGS $\xi$ that is \emph{not} perfect, the algorithm can
decompose $\xi$ into finitely many MWGS that are smaller w.r.t.\ some
lexicographical ordering.  This implies that after finitely many steps, all
remaining MWGS (if any) must be perfect. Hence, if at the very end, some
perfect MWGS remains, we can conclude the existence of a run; otherwise, we can
conclude that no run exists.

\myparagraph{Similarities} 
Let us first highlight similarities between our construction and the KLM-decomposition.
Our construction follows a similar overall strategy: While KLM decomposes MWGS, our algorithm decomposes path-schemes.
Moreover perfectness in MWGS has a correspondence in perfectness of path-schemes, since
 the  perfectness condition in~\cref{subsec:path-schemes} consists
of a (i)~``coverability part'' and (ii)~a ``linear inequalities'' part.
	Indeed, recall that a path-scheme $\rho$ is perfect if it is pre-perfect and $L(\rho)\cap
	L_{\Sigma}^{\bx,\by}\ne\emptyset$. First, pre-perfectness is closely related to coverability (of the
	zero vector) in CVAS along a run over a given word $w\in\Sigma^*$. This is because
	whether some word $w\in\Sigma^*$ can cover the zero vector (starting
	from a given initial vector) depends only on the order of first
	appearances of letters in $w$~\cite[Lemma
	3.10]{DBLP:journals/pacmpl/BalasubramanianMTZ24}. Similarly, last appearances determine ``backwards coverability''.
	Furthermore, the proof of \cref{regular-intersection-decidable} (specifically, the $\NP$ upper bound in \cite[Thm~4.9]{blondinLogicsContinuousReachability2017}) is based on expressing $L(\rho)\cap
	L_{\Sigma}^{\bx,\by}\ne\emptyset$ as solvability of systems of
	linear inequalities.

Another similarity is that KLM-based proofs
have a step that constructs a run (or a subset of runs) inside a perfect MWGS.
Examples include Lambert's Iteration
Lemma~\cite[Lemma~4.1]{DBLP:journals/tcs/Lambert92} or Leroux and Schmitz's
discovery that the pre-runs of perfect MWGS actually form an \emph{adherent
ideal} of the set of all pre-runs below valid
runs~\cite[Lemma~VII.2, Theorem~VII.5]{DBLP:conf/lics/LerouxS15}. In our proof, the
construction of a regular subset $R_\rho$ of $L_{\Sigma}^{\bx,\by}$
in~\cref{thm:gathering-lifting-runs}, which is based
on~\cite[Proposition~4.5]{blondinLogicsContinuousReachability2017}, is similar
(in spirit and in technique) to these steps in KLM-based proofs.

\myparagraph{Differences}
Let us now point out differences between our proof and existing variants of the KLM-decomposition.

First, our decomposition does not stop when all path-schemes are perfect.
Instead, it is possible that the regular sets $R_\rho$ do not cover the entire
set $L_{\Sigma}^{\bx,\by}$: Classical KLM-decompositions only construct over-
and/or underapproximations of the set of runs\footnote{
Examples of overapproximations are semilinear overapproximations (which are implicit in
the various KLM-based algorithms for reachability, including the new algorithm
for PVASS by Guttenberg, Keskin, and Meyer~\cite{PVASS}), regular
overapproximations~\cite{DBLP:conf/icalp/HabermehlMW10,DBLP:conf/icalp/CzerwinskiHZ18}, or integer VASS
overapproximations~\cite{DBLP:conf/lics/Keskin024}. Underapproximations are given, e.g. by the
Iteration Lemma~\cite[Lemma 4.1]{DBLP:journals/tcs/Lambert92}  and related
facts, see above.},   whereas we need to find an exact representation.  
This necessitates the \emph{second decomposition step}
(\cref{perfect-path-scheme-decomposition}): In addition to constructing regular
sets $R_\rho\subseteq L_{\Sigma}^{\bx,\by}$ of perfect path-schemes, we then
need to decompose the perfect path-schemes further into path-schemes that cover
those runs that have not been caught by the already-built sets $R_\rho$.

Second, the steps to achieve perfectness of path-schemes
(\cref{path-scheme-decomposition}) are very different: In KLM, one decomposes
by bounding counters or by bounding the number of transition occurrences (which
is anyway not possible in the continuous semantics). Instead, our construction
relies on a word-combinatorial observation: \cref{star-decomposition}.

Third, in order for both decomposition steps to work together (and ensure
termination), our proof requires the novel notion of gatherings, which is built
so that it (i)~yields subsets of the CVAS language (\cref{thm:gathering-lifting-runs})
and (ii)~can be achieved by decomposing arbitrary path-schemes (\cref{path-scheme-decomposition}). Note that gatherings
are different from just the set of words with a particular first- and last
appearance record (as in
\cite[Lemma~3.10]{DBLP:journals/pacmpl/BalasubramanianMTZ24}): In a gathering, it is important
that all first appearances come before all last appearances.

\section{Lower Bound}\label{sec:lower-bound}
In this section, we prove \cref{thm:non-elem}. That is, we construct a CVAS $\Sigma_h\subset\zn^{5h}$ of size $O(h)$ and configurations $\bx_n,\by_n\in\qnz^{5h}$ of size $\le n$ such that any NFA for $L_{\Sigma_h}^{\bx_n,\by_n}$ has at least $\exp_h(n)$ states. 

Let us first give an overview of the constructed CVAS\@. We will describe $\Sigma_h$ implicitly by constructing transition sequences $w_1,\ldots,w_h$, $r_1,\ldots,r_h$, and $u$ over the transition set
$\Sigma_h$. Thus, $\Sigma_h$ will just consist of those transitions occurring in $w_1,\ldots,w_h$ (which have 6 transitions each), $r_1,\ldots,r_h$ (which have one transition each), and $u$ (which has $2h$ transitions). We will be interested in transition sequences of the form
\begin{align} w_1^{\ell_1} r_1 w_2^{\ell_2} r_2 \cdots w_h^{\ell_h}r_h u\,, \label{lower-bound-runs}\end{align}
that are ``short'' in the following sense:
\begin{definition}[Short run]
A run as in \eqref{lower-bound-runs} is called \emph{short} if
$\ell_1\le n$ and $\ell_{i+1}\le \ell_i\cdot 2^{\ell_i}~\text{for $i\in[1,h-1]$}$.
\end{definition}
In fact, our construction will ensure that there is only one short run from $\bx_n$ to $\by_n$ of the form in \eqref{lower-bound-runs} and that this run will ``max out'' the bounds above.
\begin{definition}[Maxed out run]
	A short run as in \eqref{lower-bound-runs} is called \emph{maxed out} if $\ell_1=n$ and $\ell_{i+1}=\ell_i\cdot 2^{\ell_i}$ for each $i\in[1,h-1]$. 
\end{definition}
Once we define all details of $\Sigma_h$, $\bx_n,\by_n$, we will show:
\begin{lemma}\label{lower-bound-unique-short-run}
The maxed out run leads from $\bx_n$ to $\by_n$.
Moreover, the only short run of the form in \eqref{lower-bound-runs} from $\bx_n$ to $\by_n$ is the maxed out run.
\end{lemma}
In fact, \cref{lower-bound-unique-short-run} suffices to obtain the stated lower
bound:
\begin{lemma}
Every NFA for $L_{\Sigma_h}^{\bx_n,\by_n}$ has at least $\exp_h(n)$ states.
\end{lemma}
\begin{proof}
Consider an NFA for $L_{\Sigma_h}^{\bx_n,\by_n}$ with $s$ states. Since the
maxed out run leads $\bx_n$ to $\by_n$, the NFA accepts
$w_1^{\ell_1}r_1\cdots w_h^{\ell_h}r_hu$ with $\ell_1=n$ and
$\ell_{i+1}=\ell_i\cdot 2^{\ell_i}$ for $i\in[1,h-1]$. If $s<\ell_h$, then in
this run, we can remove some infix $w_h^r$ with $r>0$. This shorter word
implies the existence of a short run of the form in \eqref{lower-bound-runs} that is not maxed out, contradicting
\cref{lower-bound-unique-short-run}. Hence, we have $s\ge \ell_h\ge\exp_h(n)$.
\end{proof}

\myparagraph{Counters and Configurations}
Let us now describe the $5h$ counters and the initial and final configurations $\bx_n$ and $\by_n$.
Our CVAS will have counters
\begin{equation} \underbrace{x_1,\ldots,x_h,~~y_1,\ldots,y_h,~~\overbrace{\bar{x}_1,\ldots,\bar{x}_h,~~\bar{y}_1,\ldots,\bar{y}_h}^{\text{complement counters}}}_{\text{high-precision counters}},~~~\underbrace{\step_1,\ldots,\step_h}_{\text{step counters}}.\label{lower-bound-counter-overview}\end{equation}
Here, $\bar{x}_1,\ldots,\bar{x}_h,\bar{y}_1,\ldots,\bar{y}_h$ will act as
``complement counters'' to $x_1,\ldots,x_h,y_1,\ldots,y_h$ (they carry
complementary values). Furthermore, the $x_i,\bar{x}_i,y_i,\bar{y}_i$ for $i=1,\ldots,h$ are
``high precision counters'', since they will carry numbers of non-elementary
precision. Finally, the ``step counters'' $\step_1,\ldots,\step_h$ will
reflect the length of the run, from which we impose lower bounds on
the run length.

The initial and final configurations, $\bx_n$ and $\by_n$, are defined as
follows. Let $\bx_n$ be the configuration where $x_1,\ldots,x_h,y_1,\ldots,y_h$
all carry $\tfrac{1}{n}$, and all other counters carry $0$. Furthermore, let
$\by_n$ be the configuration where each $\step_i$ carries $4$ and all other counters
carry $0$.

\myparagraph{Shape of Maxed Out Runs} Initially, in $\bx_n$, the counters
$x_1,\ldots,x_h,y_1,\ldots,y_h$ will all be $\tfrac{1}{n}$. In the maxed out run,
the $w_1^*r_1$ portion will (overall) not change $x_1,y_1$, but it will turn
$x_2,\ldots,x_h,y_2,\ldots,y_h$ into $\tfrac{1}{n2^n}$, with $\ell_1=n$ iterations of $w_1$. Similarly, the $w_2^*r_2$
portion of this run will not change $x_1,x_2,y_1,y_2$, but it will change
$\tfrac{1}{n2^n}$ into $\tfrac{1}{n2^n\cdot 2^{n\cdot 2^n}}$ in the counters $x_{3},\ldots,x_h,y_3,\ldots,y_h$. This will be done in $w_2^{\ell_2}r_2$ with $\ell_2=n\cdot 2^{n}$. Then, this pattern continues for $i=3,\ldots,h$.

More generally, we will show that from a configuration where
$x_i,\ldots,x_h,y_i,\ldots,y_h$ all have value $\tfrac{1}{k}$, there is a
unique run $w_i^\ell r_i$ with $\ell\le k$ that arrives in a configuration where $\step_i$ is
$4$. More precisely, an \emph{$(i,k)$-short run} is one with
\begin{enumerate}
\item a transition sequence $w_i^\ell r_i$ with $\ell\le k$
\item it starts in a configuration where $x_i,\ldots,x_h,y_i,\ldots,y_h$ all
carry $\tfrac{1}{k}$, and
$\bar{x}_i,\ldots,\bar{x}_h,\bar{y}_i,\ldots,\bar{y}_h$ all carry $0$, and
$\step_i,\ldots,\step_h$ all carry $0$,
\item it ends in a configuration where $\step_i$ is $4$ and $\bar{x}_i$,
$\bar{y}_i$ carry $0$.
\end{enumerate}
We say that an $(i,k)$-short run is \emph{$(i,k)$-maxed out} if
\begin{enumerate}
\item $\ell=k$
\item at the end of the run, the counters $x_i,\ldots,x_h,y_i,\ldots,y_h$ carry $\tfrac{1}{k2^k}$, the counters $\bar{x}_i,\ldots,\bar{x}_h$, $\bar{y}_i,\ldots,\bar{y}_h$, and $\step_{i+1},\ldots,\step_{h}$ all carry $0$.
\end{enumerate}
Once we define all the transitions, we will prove:
\begin{lemma}\label{lower-bound-unique-local}
For every $i\in[1,h]$ and $k\ge 0$, there is an $(i,k)$-maxed out run. Moreover, every $(i,k)$-short run is $(i,k)$-maxed out.
\end{lemma}

\myparagraph{Transition Notation}
To make the effects of the transitions easily visible, we introduce some
notation.  Recall that our CVAS has $5h$ counters (see \eqref{lower-bound-counter-overview}). We use the names of a
counter to indicate the effect of incrementing this counter. For example, $x_i$
is the vector in $\zn^{5h}$ that has $1$ in the coordinate for counter $x_i$.
This also means, e.g., that $x_i+\step_i$ would be a transition that adds $1$
to the counters $x_i$ and $\step_i$, but leaves all others unchanged.
For example, the transition sequence $u$ is
\begin{align*}
u=f_1g_1\cdots f_hg_h, && f_j=-x_j, && g_j=-y_j,
\end{align*}
for every $j=1,\ldots,h$. Thus, if the counters
$x_1,\ldots,x_h,y_1,\ldots,y_h$ carry values $\le 1$, the sequence $u$ allows us to reset
them to zero.

The transition sequences $w_1,r_1,\ldots,w_h,r_h$
will have the following property:
\begin{property}\label{lower-bound-no-effect}
The sequences $w_i$ and $r_i$ have no effect on $x_j,\bar{x}_j$, $y_j$, $\bar{y}_j$, $\step_j$ for $j<i$.
\end{property}
Before we describe the transition sequences
$w_1,r_1,\ldots,w_h,r_h$ and prove \cref{lower-bound-unique-local}, let us see
how \cref{lower-bound-unique-local} implies
\cref{lower-bound-unique-short-run}.
\begin{proof}[Proof of \cref{lower-bound-unique-short-run}]
First, by initially firing a $(1,n)$-maxed out run, then a $(2,n2^n)$-maxed out
	run, etc., and then resetting all $x_j,y_j$ (for $j=1,\ldots,h$) using $u$, we obtain a maxed out run from $\bx_n$ to $\by_n$. Note that here, $u$ is able to reset all $x_j,y_j$ since their values are $\le 1$.

For the second statement, take a short run along the transition sequence
$w_1^{\ell_1}r_1\cdots w_h^{\ell_h}r_hu$.  Then the configuration reached after
the prefix $w_1^{\ell_1}r_1$ must agree with $\by_n$ in the counters
$\step_1,\bar{x}_1,\bar{y}_1$, because the latter are not touched during
$w_2^{\ell_2}r_2\cdots w_h^{\ell_h}r_hu$. Hence, the run along
$w_1^{\ell_1}r_1$ is a short $(1,n)$-run. By \cref{lower-bound-unique-local},
the run along $w_1^{\ell_1}r_1$ is a $(1,n)$-maxed out run. Thus, after
$w_1^{\ell_1}r_1$, our run reaches a configuration where
$x_1,\ldots,x_h,y_1,\ldots,y_h$ carry $\tfrac{1}{n2^n}$, the counters
$\bar{x}_1,\ldots,\bar{x}_h,\bar{y}_1,\ldots,\bar{y}_h$ carry $0$, and
$\ell_1=n$.  By repeating this argument, we get that $w_i^{\ell_i}r_i$ is a
$(i,M_i)$-short run, and hence a $(i,M_i)$-maxed out run, for $M_1=n$ and
$M_{i+1}=M_i2^{M_i}$ for $i\in[1,h-1]$. In particular, we obtain $\ell_i=M_i$
for all $i$. Overall, this implies that our run along $w_1^{\ell_1}r_1\cdots
w_h^{\ell_h}r_hu$ is a maxed out run.
\end{proof}

It thus remains to set up $w_1,\ldots,w_h,r_1,\ldots,r_h$ so as to satisfy
\cref{lower-bound-unique-local,lower-bound-no-effect}.


\myparagraph{The Three Gadgets}
When describing the transition sequences $w_1,\ldots,w_h$ and $r_1,\ldots,r_h$,
some more notation will be useful: Since $\bar{x}_i$ is meant as a complement
counter to $x_i$, most transitions will have opposite effect on these two counters
(and likewise with $\bar{y}_i$ and $y_i$). Therefore, we define
$\hat{x}_i:=x_i-\bar{x}_i$ and $\hat{y}_i:=y_i-\bar{y}_i$. Thus, a
transition with effect $\hat{x}_i$ adds $1$ to $x_i$ and subtracts $1$ from
$\bar{x}_i$.

Let us now describe the words $w_1,\ldots,w_h,r_1,\ldots,r_h$, together with the $(i,k)$-maxed out runs. We will show afterwards that the described $(i,k)$-maxed out run is indeed the only $(i,k)$-short run. First, we have $w_i=t_{i,1}\cdots t_{i,6}$, where the first three transitions are:
\begin{align*}
&t_{i,1}=-2\hat{x}_i-\hat{y}_i-2\sum_{j=i+1}^h (\hat{x}_j+\hat{y}_j), &&
t_{i,2}=\hat{x}_i+\step_i, &&
t_{i,3}=-\hat{x}_i+\step_i.
\end{align*}
Thus, $t_{i,1}$ subtracts some value $\alpha\in(0,1]$ from $y_i$, and subtracts $2\alpha$ from $x_i$ and all counters $x_{i+1},y_{i+1},\ldots,x_h,y_h$. In our $(i,k)$-maxed out run, this is applied with $\alpha=\tfrac{1}{2k}$: This depletes the counters $x_i$ and $x_{i+1},y_{i+1},\ldots,x_h,y_h$, but sets $y_i$ to $\tfrac{1}{2k}$.
Because of opposite effects on the complement counters, this sets $\bar{x}_i$ and $\bar{x}_{i+1},\bar{y}_{i+1},\ldots,\bar{x}_h,\bar{y}_h$ to $\tfrac{1}{k}$. Moreover, $\bar{y}_i$ becomes $\tfrac{1}{2k}$.

In $t_{i,2}t_{i,3}$, we move mass back and forth between $x_i$ and $\bar{x}_i$, while adding the mass (twice) to $\step_i$.
In our $(i,k)$-maxed out run, $t_{i,2}t_{i,3}$ moves all the mass
($\tfrac{1}{k}$) from $\bar{x}_i$ to $x_i$ (in $t_{i,2}$), and then back (in
$t_{i,3}$); and we add $\tfrac{2}{k}$ to $\step_i$. Thus, compared to the
configuration after $t_{i,1}$, we only added $\tfrac{2}{k}$ to $\step_i$.

The second part of $w_i$, namely $t_{i,4}t_{i,5}t_{i,6}$, works as follows:
\begin{align*}
&t_{i,4}=\sum_{j=i}^h \hat{x}_j+\hat{y}_j &&
t_{i,5}=-\hat{y}_i+\step_i &&
t_{i,6}=\hat{y}_i+\step_i
\end{align*}
In $t_{i,4}$, we fill the counters $x_i,y_i,\ldots,x_h,y_h$ back up, but only as far as the capacity of $y_i$ (as guaranteed using the complement counter $\bar{y}_i$) allows: In our $(i,k)$-maxed out run, before $t_{i,4}$, we have $y_i$ at $\tfrac{1}{2k}$, and the counters $x_i$ and $x_{i+1},y_{i+1},\ldots,x_h,y_h$ at $0$. Thus, with $t_{i,4}$, we can increase all counters $x_i,y_i,\ldots,x_h,y_h$ by $\tfrac{1}{2k}$. This means, we get $\tfrac{1}{2k}$ in $x_i$ and in $x_{i+1},y_{i+1},\ldots,x_h,y_h$. And we get $\tfrac{1}{k}$ in $y_i$.

Along $t_{i,5}t_{i,6}$, we then have a ``back and forth'' again, but now between $y_i$ and $\bar{y}_i$. In our $(i,k)$-maxed out run, we move the entire $\tfrac{1}{k}$ back and forth. We thus have the same configuration as after $t_{i,4}$, but $\step_i$ now carries $\tfrac{2}{k}+\tfrac{2}{k}=\tfrac{4}{k}$.

We have now described our $(i,k)$-maxed out run along the first copy of $w_i$. Note that overall, it had the same effect on $x_i$ and on $x_{i+1},y_{i+1},\ldots,x_h,y_h$: It replaced the value $\tfrac{1}{k}$ with the value $\tfrac{1}{2k}$. Moreover, the run did not change $y_i$. And it added $\tfrac{4}{k}$ to $\step_i$. To describe the later iterations of $w_i$, we use a simpler notation for configurations: The abovementioned configuration is written as
\vspace{0.6cm}
\begin{equation} ~~~~~~~~~~~~~~~~~~~~~~~~~~~~~~~~~~~~~~~~~~~\left\langle\eqnmarkbox{node1}{\frac{1}{2k}};\eqnmarkbox{node2}{\frac{1}{k}};\eqnmarkbox{node3}{\frac{4}{k}}\right\rangle \label{lower-bound:notation-configurations}\end{equation}
\vspace{0.5cm}
\annotate[yshift=1em]{left}{node1}{content of $x_i$ and of $x_{i+1},y_{i+1},\ldots,x_h,y_h$}
\annotate[yshift=1em]{}{node2}{content of $y_i$}
\annotate[yshift=-1em]{below}{node3}{content of $\step_i$}

(Here, we use semicolons instead of commas to distinguish this representation from a three-dimensional configuration.)
In the $j$-th iterations of $w_i$ for $j\ge 1$, our $(i,k)$-maxed out run will perform the following transformation:
\begin{multline}
\left\langle\frac{1}{2^jk}; \frac{1}{k}; \frac{4j}{k}\right\rangle \xrightarrow{t_{i,1}} \left\langle 0;\frac{1}{k}-\frac{1}{2^{j+1}k};\frac{4j}{k}\right\rangle\xrightarrow{t_{i,2}t_{i,3}}\left\langle 0;\frac{1}{k}-\frac{1}{2^{j+1}k};\frac{4j+2}{k}\right\rangle \\
\xrightarrow{t_{i,4}t_{i,5}t_{i,6}}\left\langle\frac{1}{2^{j+1}k};\frac{1}{k};\frac{4(j+1)}{k}\right\rangle\label{lower-bound-behavior-maxed-out}
\end{multline}
and overall we went from $\langle\tfrac{1}{2^jk}; \tfrac{1}{k}; \tfrac{4j}{k}\rangle$ to $\langle\tfrac{1}{2^{j+1}k}; \tfrac{1}{k}; \tfrac{4(j+1)}{k}\rangle$.
Thus, after $w_i^k$, our $(i,k)$-maxed out run arrives in $\langle\tfrac{1}{2^kk};\tfrac{1}{k};4\rangle$.
After executing $w_i^k$, our $(i,k)$-maxed out run has the reset transition $r_i$:
\[ r_i = -\bar{x}_i - \sum_{j=i+1}^h (\bar{x}_j+\bar{y}_j). \]
This allows us to reset all the complement counters to $0$, ready for
the next stage of $w_{i+1}$. This completes our description of
$w_1,r_1,\ldots,w_h,r_h$ and of our $(i,k)$-maxed out run. Moreover,
\cref{lower-bound-no-effect} is clearly satisfied. For
\cref{lower-bound-unique-local}, it remains to show that this run is the only
$(i,k)$-short run.

\begin{proof}[Proof of \cref{lower-bound-unique-local}]
Consider an $(i,k)$-short run along $w_i^\ell r_i$. The transitions in $w_i$ keep constant both (i)~the sum of $x_i$ and $\bar{x}_i$ and (ii)~the sum of $y_i$ and $\bar{y}_i$. These two sums remain $\tfrac{1}{k}$ throughout the execution of $w_i^\ell$. This means, the transitions $t_{i,2}$, $t_{i,3}$, $t_{i,5}$, and $t_{i,6}$ can use a firing fraction of at most $\tfrac{1}{k}$. We call these the ``$\step_i$-transitions'', because they modify $\step_i$. Hence, if we denote by $s$ the sum of all firing fractions of $\step_i$-transitions, then $s\le \ell\cdot4\cdot \tfrac{1}{k}$.

Moreover, $s$ is the final value in $\step_i$, which is $4$. Hence $4=s\le \ell\cdot 4\cdot\tfrac{1}{k}$. Since $\ell\le k$, this is only possible if $\ell=k$ and all fractions of $\step_i$-transitions are $\tfrac{1}{k}$. But, to fire $t_{i,2}$ with fraction $\tfrac{1}{k}$, the $t_{i,1}$ directly before must leave $x_i$ empty. Likewise, to fire $t_{i,5}$ with fraction $\tfrac{1}{k}$, the $t_{i,4}$ directly before must leave $y_i$ at full capacity, i.e.~$\tfrac{1}{k}$. In short, \emph{each $t_{i,1}$ must empty $x_i$} and \emph{each $t_{i,4}$ must fill $y_i$ to $\tfrac{1}{k}$}.

By induction on $j$, this implies that the $j$-th iteration of $w_i$ must have exactly the effect described in \eqref{lower-bound-behavior-maxed-out}. In particular, after $\ell=k$ iterations, we end up in the configuration $\langle \tfrac{1}{2^kk}, \tfrac{1}{k}, 4\rangle$. This implies that the counters $\bar{x}_i$ and $\bar{x}_{i+1},\bar{y}_{i+1},\ldots,\bar{x}_h,\bar{y}_h$ all carry $\tfrac{1}{k}-\tfrac{1}{2^kk}$. After $w_i^\ell=w_i^k$, the transition $r_i$ must be fired with some fraction $\alpha\in(0,1]$. After this, the counter $\bar{x}_i$ will carry $\tfrac{1}{k}-\tfrac{1}{2^kk}-\alpha$. Since $\bar{x}_i$ is $0$ in the final configuration of our $(i,k)$-short run, we must have $\alpha=\tfrac{1}{k}-\tfrac{1}{2^kk}$. In particular, all the counters $\bar{x}_i$ and $\bar{x}_{i+1},\bar{y}_{i+1},\ldots,\bar{x}_h,\bar{y}_h$ are reset to $0$ by $r_i$. This means, our $(i,k)$-short run was in fact an $(i,k)$-maxed out run.
\end{proof}

\begin{remark}
	It should be noted that the requirement of being short (by our
	definition of ``short'') is crucial to force the non-elementary length
	of the run $w_1^{\ell_1}r_1\cdots w_h^{\ell_h}r_hu$. Indeed, the proof
	of \cite[Thm.~4.9]{blondinLogicsContinuousReachability2017} implies
	that in every CVAS with states, if there is a run between two
	configurations, then there is one of at most exponential length.  In
	particular, there must be such a run from $\bx_n$ to $\by_n$ of the form
	\eqref{lower-bound-runs}. Indeed, one can choose
	$\ell_i=2^in$ for each $i=1,\ldots,h$ (and since $2n>n$, this is not short by our definition). See
	\cref{app:lower-bound} for details.
\end{remark}

\section{Conclusion}
We have shown that the language of every CVAS is effectively regular (\cref{th:main}). This clarifies the decidability landscape of reachability in systems with continuous counters (\cref{intersection-decidable}). While the decidability aspects are now understood, we view these results as the starting point for investigating the complexity of reachability in a wide range of systems with continuous counters.

\myparagraph{Complexity} There are many complexity aspects that remain unclear.
For example, what is the exact complexity of computing NFAs for
$L_{\Sigma}^{\bx,\by}$? We have a non-elementary lower bound, and our
construction yields an Ackermannian upper bound. It would be interesting to investigate whether there is a primitive-recursive construction.

We already know situations where computing an NFA for
$L_{\Sigma}^{\bx,\by}$ (which has non-elementary size) does not yield optimal
complexity for reachability in systems with continuous counters. For example,
in finite-state systems equipped with continuous counters, reachability is
$\NP$-complete~\cite[Theorem~4.14]{blondinLogicsContinuousReachability2017}, and for pushdown
systems with continuous counters, reachability is
$\NEXP$-complete~\cite[Theorem~1.1]{DBLP:journals/pacmpl/BalasubramanianMTZ24}. This
suggests that for order-$k$ pushdown systems with continuous counters,
reachability should be $k$-$\NEXP$-complete, but our results only yield an
Ackermannian upper bound.

\myparagraph{Demystification} It would also be interesting to provide a
simple, abstract description of how the path-scheme decomposition in our
construction relates to the CVAS. For the KLM decomposition, such a
description was obtained by Leroux and Schmitz~\cite{DBLP:conf/lics/LerouxS15}.
Their description---the so-called ``demystification''---casts the
decomposition as a computation of ideals. However, this was 30 years after
the KLM decomposition was first described. Hence, this might require deep
conceptual insights. 

For example, we can ask if there is a counterpart to the run embedding
introduced by Jan\v{c}ar~\cite{DBLP:journals/tcs/Jancar90} and
Leroux~\cite{DBLP:conf/popl/Leroux11} that would make our construction an ideal
decomposition. Also, since our regular language is a finite union of sets that
look similar to upward closures: Is there a well-quasi ordering such that
$L_{\Sigma}^{\bx,\by}$ is upward closed? A potential candidate for these two
questions might be the DFA-defined well-quasi orderings from
\cite[p.~3]{DBLP:conf/lics/Zetzsche18}, where the DFA is induced by a suitable
pre-perfect path-scheme.

On the one hand, such a demystification would conceptually clarify our proof.
On the other hand, it might also deepen our understanding of VAS: Because of
the alternation of decomposition steps, our construction yields a precise
description of firing sequences, whereas the KLM decomposition for VAS only
provides an overapproximation (i.e.\ an ideal decomposition over the set of
pre-runs). Perhaps the two approaches can be merged into a more general construction
for VAS that lets us understand the set of runs of a VAS as a regular set,
once runs are equipped with appropriate structure --- just as context-free
languages are yields of regular sets of trees.

\section*{Acknowledgments}
We thank the anonymous reviewers for their suggested improvements.

This research was sponsored in part by the Deutsche Forschungsgemeinschaft project number 389792660 TRR 248--CPEC.

\raisebox{-8pt}[0pt][0pt]{\includegraphics[height=0.8cm]{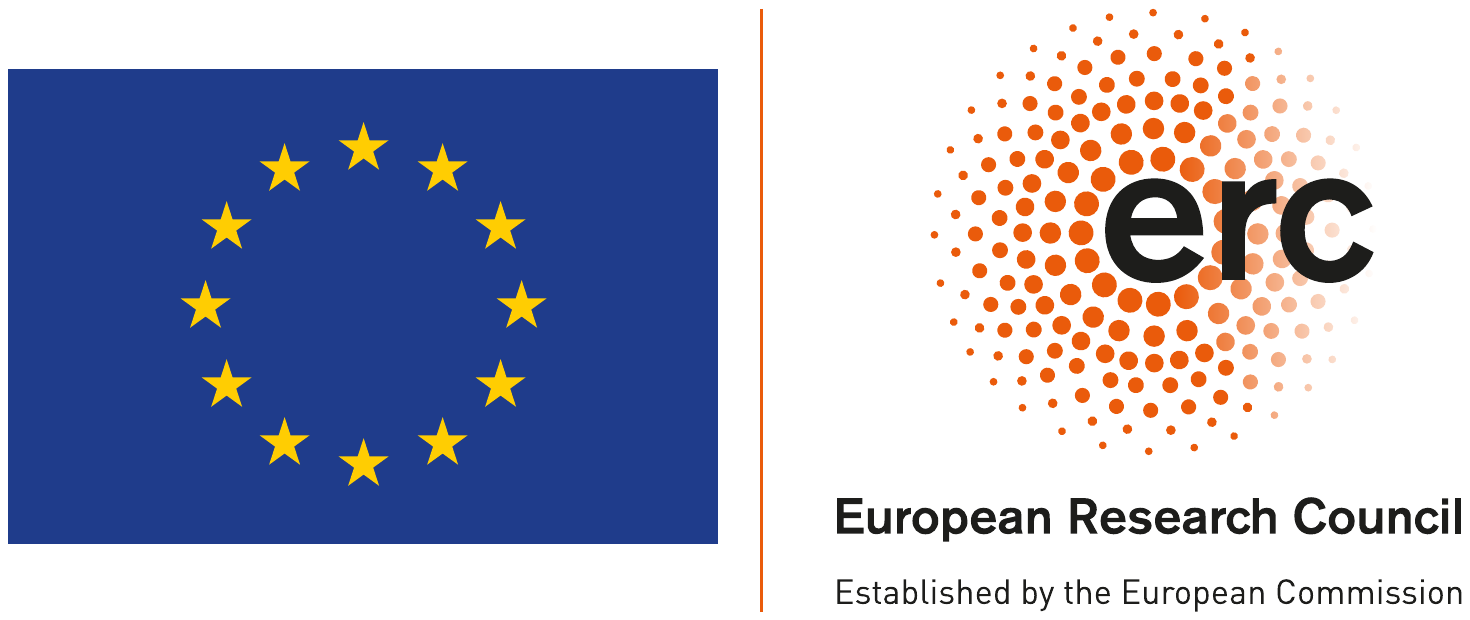}}
Funded by the European Union (ERC, FINABIS, 101077902).
Views and opinions expressed are however those of the authors only and do not necessarily reflect those of the European Union
or the European Research Council Executive Agency.
Neither the European Union nor the granting authority can be held responsible for them.

\label{beforebibliography}
\newoutputstream{pages}
\openoutputfile{main.pages.ctr}{pages}
\addtostream{pages}{\getpagerefnumber{beforebibliography}}
\closeoutputstream{pages}
\bibliographystyle{plainurl}
\bibliography{refs}
\label{endbibliography}
\newoutputstream{pageendbibliography}
\openoutputfile{main.pageendbibliography.ctr}{pageendbibliography}
\addtostream{pageendbibliography}{\getpagerefnumber{endbibliography}}
\closeoutputstream{pageendbibliography}

\newpage
\label{startappendix}
\newoutputstream{pagestartappendix}
\openoutputfile{main.pagestartappendix.ctr}{pagestartappendix}
\addtostream{pagestartappendix}{\getpagerefnumber{startappendix}}
\closeoutputstream{pagestartappendix}

\appendix

\crefalias{section}{appsec}
\crefalias{subsection}{appsec}

\section{Complexity}\label[appsec]{app:complexity}
\newcommand{\Fast}{\mathbf{F}}
\newcommand{\Func}{\mathscr{F}}
Let us briefly observe that our construction runs in at most Ackermannian
time~\cite{DBLP:journals/toct/Schmitz16}. To be precise, by this we mean that an NFA
can be computed in a time bound used for the definition of the class $\Fast_\omega$ of Ackermannian problems.
(Strictly speaking, $\Fast_\omega$ is a class of decision problems, so we can
technically not say that we can compute the NFA in $\Fast_\omega$.)

\subsection{An Ordinal Ranking Function} To begin with, note that on every path of our
tree, the weight vectors strictly decrease when comparing an odd level with
the next odd level.
Moreover, each decomposition of a path-scheme $\rho$ clearly
produces new path-schemes whose size is elementary\footnote{This is obvious from the descriptions in this paper, except perhaps when we use enumeration (in the proof of \cref{thm:lifting-runs}) to find a tuple of vectors that can sustain a run. However, the vectors in this tuple can also be bounded, even exponentially: This is because finding these vectors can be phrased as reachability in a CVAS with states (CVASS), and the reachability algorithm of Blondin and Haase~\cite[Theorem 4.9]{blondinLogicsContinuousReachability2017} for CVASS yields an exponentially long path with fractions that require polynomially many bits, if there is one.} in the size of $\rho$ and
the size of the vectors $\bx$ and $\by$.\footnote{By size, we mean the number of bits needed to write down a path-scheme or rational vector, respectively.}
This means, there is an elementary
function $f\colon\nn\to\nn$ such that for the sequence $\rho_1,\rho_2,\ldots$
of path-schemes occurring on odd levels of a tree branch, we have
$|\rho_{i+1}|\le f(|\rho_i|+|\bx|+|\by|)$. This implies that for some
elementary $g\colon\nn\to\nn$, we have $|\rho_i|\le g^i(n)$, where $n$ is an
upper bound for both the size of our input VAS and the size of $\bx$ and $\by$.
Here, $g^i\colon\nn\to\nn$ is the $i$-fold iteration of $g$.
Similarly, we may assume that the NFA constructed for the perfect path-scheme
$\rho_i$ has at most $g^i(n)$-many states.  Therefore, we have
$|\weight{\rho_i}|\le g^i(n)$ for every $i\ge 1$.

\subsection{Bounding Path Lengths}
The discussion from the previous paragraph implies that the weight function $\weight{\cdot}$
is what is called a \emph{$(g,n)$-controlled ranking function}, and by a length
function theorems for ordinals~\cite[Theorem 3.5]{DBLP:books/hal/Schmitz17}
this implies that any branch of the tree has at most $g_{\omega^{|\Sigma|}}(n)$
odd-level nodes.
Here, for an ordinal $\alpha$, the function
$g_\alpha\colon\nn\to\nn$ is the $\alpha$-th Cicho\'{n} function based on
$g$~\cite[Section 3.2.2]{DBLP:books/hal/Schmitz17}.

\subsection{Bounding The Run Time}
We have thus bounded the length of the longest path in our tree. Let us now use
this to estimate the running time of our construction. We expand our tree
level-by-level. Suppose it takes time $t(i)$ to construct levels $1,\ldots,2i$
of our tree, including an NFA for the union of all subsets of $L_{\Sigma}^{\bx,\by}$ encountered so far. Then similar to the observation above, there is an elementary
function $h\colon\nn\to\nn$ such that $t(i)\le h^i(n)$. Without loss of generality, we can pick the elementary function $g$ above so that it also provides the time upper bound. Thus, we have $t(i)\le g^i(n)$.

Since the longest path in our tree has at most $g_{\omega^{|\Sigma|}}(n)$-many odd-level nodes, the full construction of our finite tree (including the final NFA) takes time
\[ t(g_{\omega^{|\Sigma|}}(n))\le g^{g_{\omega^{\Sigma}}(n)}(n). \]
By~\cite[Eq.~3.14]{DBLP:books/hal/Schmitz17}, the latter expression is exactly
\[ g^{\omega^{|\Sigma|}}(n), \]
where for an ordinal $\alpha$, the function
$g^\alpha\colon\nn\to\nn$ is the $\alpha$-th Hardy function based on
$g$~\cite[Section 3.2.2]{DBLP:books/hal/Schmitz17}. Thus, our construction can
be performed in time $g^{\omega^{|\Sigma|}}(n)$.
Hence, if $h=|\Sigma|$ is fixed, then we obtain a run-time bound of $g^{\omega^h}(n)$. And if the size of $\Sigma$ is not fixed, then we obtain a run-time bound of $g^{\omega^{\omega}}(n)$.

A \emph{decision problem} that runs in time $g^{\omega^{h}}(n)$ (resp.\
$g^{\omega^\omega}(n)$) on input size $n$ belongs to the relativized class
$\Fast_{g,h}$ (resp.\ $\Fast_{g,\omega}$). Since $g$ belongs to the
class $\Func_3$ of elementary functions, the class $\Fast_{g,|\Sigma|}$ (resp.\
$\Fast_{g,\omega}$) is included in $\Fast_{h+4}$ (resp.\
$\Fast_{\omega+4}=\Fast_{\omega}$) by \cite[Theorem
4.6]{DBLP:books/hal/Schmitz17}.  Here, we are not considering a decision
problem, but one of computing a function. However, the proof of \cite[Theorem
4.6]{DBLP:books/hal/Schmitz17}, i.e.\ that $\Fast_{g,\alpha}$ is included in
$\Fast_{\alpha+4}$ for elementary functions $g$, applies verbatim to complexity
classes for computing functions, since the proof (which can be found in
\cite[Theorem 4.2]{DBLP:journals/toct/Schmitz16}) is only about estimating time
bounds. This implies that for fixed $h=|\Sigma|$, we can compute the NFA in the
time bound for the class $\Fast_{h+4}$ (i.e.\ in primitive-recursive time), and
when $|\Sigma|$ is not fixed, then we can compute the NFA in the time bound for
$\Fast_\omega$ (i.e.\ Ackermannian time).

%
%
%
%

%


\section{Completion of Proof of Lemma~\ref{star-decomposition}}\label[appsec]{app:star-decomposition}

In this section, we complete the proof of~\cref{star-decomposition} by completing the proof of the induction step for $X_A$ when $|A| > 0$. Recall that we have shown that, in this case, $A^*$ can be decomposed as
\begin{equation}\label{eq:app-path-scheme-decomposition}
	A^* ~~~=~~~ \bigcup_G L(G)~~\cup~~\bigcup_{\begin{smallmatrix}B\subsetneq A \\ C\subsetneq A\end{smallmatrix}} B^*C^*~~\cup~~\bigcup_{\begin{smallmatrix}a\in A,\\B\subseteq A\setminus\{a\},\\C\subseteq A\setminus\{a\}\end{smallmatrix}} B^*aC^*,	
\end{equation}
where $G$ ranges over all gatherings $X_{a_1,\ldots,a_n}^{b_1,\ldots,b_n}$ with $A=\{a_1,\ldots,a_n\}=\{b_1,\ldots,b_n\}$. With~\cref{eq:app-path-scheme-decomposition} at our disposal, we now prove the following claim, which suffices to complete the proof of the induction step for $X_A$:
For each term $S$ that appears in the RHS of~\cref{eq:app-path-scheme-decomposition}, we can decompose $S$ as a finite set of pre-perfect path-schemes $\sigma_1,\ldots,\sigma_k$ such that each $\weight{\sigma_i} \lexeq \weight{X_A}$.
We do this by a case analysis on $S$.

\begin{itemize}
	\item Suppose $S = L(G)$ for some gathering $G = X_{a_1,\ldots,a_n}^{b_1,\ldots,b_n}$ with $A=\{a_1,\ldots,a_n\}=\{b_1,\ldots,b_n\}$. By definition each gathering is a pre-perfect path-scheme itself. Furthermore, by definition, we have that $\weight{G} = \weight{A}$ and so we are done.
	\item Suppose $S = B^*C^*$ for some $B \subsetneq A, C \subsetneq A$. In this case, we
	use the induction hypothesis on $X_B$ to get a finite set of pre-perfect path-schemes
	$\rho_1^B,\dots,\rho_m^B$ such that $L(X_B) = B^* = \bigcup_{i} \rho_i^B$ and for each $i$, $\weight{\rho_i^B} \lexeq \weight{X_B}$. Similarly, we get a finite set of pre-perfect path-schemes
	$\rho_1^C,\dots,\rho_n^C$ for $X_C$. It is then clear that 
	$S = \bigcup_{\begin{smallmatrix}1 \le i \le m, \\ 1 \le j \le n\end{smallmatrix}} \rho_i^B \rho_j^C$.
	Furthermore $\weight{\rho_i^B\rho_j^C} = \weight{\rho_i^B} + \weight{\rho_j^C} \lexeq 
	\weight{X_B} + \weight{X_C} = \overline{|B|} + \overline{|C|} \lexlt \overline{|A|} = \weight{X_A}$.
	Here all the equalities follow by definition of the weight function, the first inequality follows from induction hypothesis and the second inequality follows
	from the fact $\overline{|B|} + \overline{|C|}$ comprises only 0's in any component bigger than or equal to $|A|$. This completes the proof for this case.
	\item Suppose $S = B^*aC^*$ for some $a \in A, B \subsetneq A, C \subsetneq A$. Similar to the
	previous case, we use the induction hypothesis on $X_B$ and $X_C$ to get a finite set of pre-perfect path-schemes
	$\rho_1^B,\dots,\rho_m^B$ and
	$\rho_1^C,\dots,\rho_n^C$. It is then clear that 
	$S = \bigcup_{\begin{smallmatrix}1 \le i \le m, \\ 1 \le j \le n\end{smallmatrix}} \rho_i^B a \rho_j^C$.
	Also, the fact that $\weight{\rho_i^Ba\rho_j^C} \lexeq \weight{X_A}$ follows by the same reasoning as the previous case. This completes the proof for this case.
\end{itemize}

This establishes that $X_A$ can be decomposed as a finite union of pre-perfect path schemes such that the weight of each such path-scheme is at most the weight of $X_A$.
Hence, this completes the proof of the induction step and also allows us to conclude the proof of the lemma.

\section{Completion of Proof of Theorem~\ref{thm:gathering-lifting-runs}}\label[appsec]{app:liftingruns}

In this section, we complete the proof of~\cref{thm:gathering-lifting-runs}, by constructing 
firing sequences $\alpha', \beta'$ and $\gamma'$
such that $\bx \xrightarrow{\alpha' u' }\bx' \xrightarrow{\beta' \centr(w')} \by' \xrightarrow{\gamma' v'} \by$. To this end, recall that we already have runs $r_1 = \bx \xrightarrow{\alpha u} \bx'$, $r_2 = \bx' \xrightarrow{\beta \centr(w)} \by'$ and $r_3 = \by' \xrightarrow{\gamma v} \by$ over the words $u, \centr(w)$ and $v$
such that $u \preceq u'$, $\centr(w) \preceq \centr(w')$, $v \preceq v'$ and the following properties are satisfied:
\begin{itemize}
	\item If a counter is non-zero at any point in the run $r_1$,
	then it stays non-zero from that point onwards in $r_1$.
	\item If a counter is zero at any point in the run $r_3$,
	then it stays zero from that point onwards in $r_3$.
	\item If a counter was non-zero at any point in either $r_1$ or $r_3$,
	then it stays non-zero along the run $r_2$.
\end{itemize}
We now use these runs $r_1, r_2$ and $r_3$ to construct the required firing sequences $\alpha' u', \beta' \centr(w')$ and $\gamma' v'$. We begin by first constructing $\alpha'u'$.

\subsection*{Constructing $\ \alpha'u'$} 

We now describe how we can redistribute the firing fractions in $\alpha u$ to accommodate for the 
additional occurrences of each transition in $u'$.

First, note that $u = a_1a_2\ldots a_n$ and $u'$ is some word which is in the language defined by
\[
    a_1 \{a_1\}^* a_2 \{a_1,a_2\}^* \dots a_{n-1}\{a_1,\dots,a_{n-1}\}^* a_n \ .
\]
Hence, we can split $u'$ as $u' = a_1 u_2 a_2 u_3 \dots a_{n-1} u_n a_n$
where each $u_i \in \{a_1,\dots,a_{i-1}\}^*$.
For each letter $a_i$, let $\extra_i$ be the number of times $a_i$ appears in the word $u_2 u_3 \dots u_n$,
i.e., $\extra_i$ is the number of times $a_i$ appears in $u'$ minus 1.

Let us denote the run $r_1 = \bx \act{\alpha u} \bx'$ by
$\bx = \bx_0 \xrightarrow{\alpha_1 a_1} \bx_1 \xrightarrow{\alpha_2 a_2} \bx_2 \ldots \bx_{n-1} \xrightarrow{\alpha_n a_n} \bx_n = \bx'$.
By assumption on $r_1$, we are guaranteed that for any $1 \le i \le n$ and any counter $c$, if $\bx_i(c) > 0$ 
then $\bx_j(c) > 0$ for all $j \ge i$. Since we can only decrement a counter if it is non-zero,
this means that 
\begin{quote}
	\textbf{Fact 1: } For any counter $c$ that is decremented by the transition $a_i$,
	we have $\bx_{i}(c) = \bx(c) + \sum_{j=1}^{i} \alpha_j a_j(c) > 0$.
	Hence, it also then follows that $\bx_{j}(c) > 0$ for any $j \ge i$.
\end{quote}

Let $\epsilon$ be a very small non-zero fraction such that for every $1 \le i \le n$,
$\extra_i \cdot \epsilon < \alpha_i$ and $(\sum_{k=1}^n \extra_k) \cdot \epsilon |a_j(c)| < \bx_i(c)$ for every $\bx_i(c) > 0$ and every $a_j$.
For each $i$, we now let $\alpha''_i = \alpha_i - \epsilon \cdot \extra_i$.
Now, in the word $u'$, for each $a_i$, we are going to fire the first appearance of $a_i$
with fraction $\alpha_i''$ and every other occurrence by $\epsilon$. Let $\alpha' u'$ be the
firing sequence obtained by firing each transition in $u'$ as specified in the previous line.

If we now prove that while firing $\alpha' u'$ from $\bx$, we never go below zero on any counter,
then for each counter $c$, the final value is $\bx(c) + \sum_{i=1}^n (\alpha''_i + \extra_i \cdot \epsilon) a_i(c) = \bx(c) + \sum_{i=1}^n \alpha_i a_i(c) = \bx'(c) $.
So, if we never go below zero on any counter, this is a valid run from $\bx$ to $\bx'$.

So the only thing left to prove is that while firing $\alpha' u'$ from $\bx$, no counter goes below zero. Towards a contradiction, suppose the value of some counter $c$ becomes negative at some point. Pick the first such point. 
Let $a_1 u_2 \dots a_{i-1} u_i u''$ be the run till that point such that $u''$ is a prefix
of $a_i u_{i+1}$.

\textbf{Case 1: } Suppose $u'' = a_i$. Note that till the last point of $u_{i}$, for each $j \le i-1$, we have fired $a_j$ with fraction $\alpha_j''$
once and then again for some number of times (say $h_j \le \extra_j$) with fraction $\epsilon$.
Hence, the value of the component $c$ till the last point of $u_i$ is
$V = \bx(c) + \sum_{j=1}^{i-1} (\alpha_j'' + h_j  \cdot\epsilon ) a_j(c) \ge 0$.
Since $u''$ is the first point at which $a_i$ occurs, $a_i$ is fired with fraction $\alpha_i''$.
Now since $c$ becomes negative after firing $a_i$ with fraction $\alpha_i''$, it must be the case that $a_i$ decrements $c$ and $V + \alpha_i'' a_i(c) < 0$.

By Fact 1, it follows that $\bx_{i}(c) = \bx(c) + \sum_{j=1}^{i} \alpha_j a_j(c) > 0$. 
Hence, we then have
\begin{align*}
	&V + \alpha_i'' a_i(c) \\
    &\quad = \bx(c) + \sum_{j=1}^{i-1} (\alpha_j'' + h_j \cdot \epsilon) a_j(c) +  \alpha_i'' a_i(c) & \text{ Substitute value for } V\\
	&\quad= \bx_i(c) - \sum_{j=1}^i \alpha_j a_j(c) + \sum_{j=1}^{i-1} (\alpha_j'' + h_j \cdot \epsilon) a_j(c) + \alpha_i'' a_i(c) & \text{ Substitute value for } \bx(c)\\
	&\quad= \bx_i(c) -\sum_{j=1}^{i-1} (\extra_j - h_j) \cdot \epsilon a_j(c) - \extra_i \cdot \epsilon a_i(c) & \text { Substitute value for } \alpha_j''\\
	&\quad\ge \bx_i(c) - \left (\sum_{j=1}^i \extra_j \right) \cdot \epsilon (\max\{|a_1(c)|, |a_2(c)|, \ldots, |a_i(c)|\}) 
\end{align*}

Now since $\bx_i(c) > 0$, by definition of $\epsilon$, it follows that the last 
value is $> 0$ and so we have that $V + \alpha_i'' a_i(c) > 0$, which leads to a contradiction.


\textbf{Case 2: } Suppose $u'' \neq a_i$, but is a prefix of $a_i u_{i+1}$. Hence, till the penultimate point of $u''$, for each $j \le i$, we have fired $a_j$ with fraction $\alpha_j''$
once and then again for some number of times (say $h_j \le \extra_j$) with fraction $\epsilon$.
Hence, the value of the component $c$ till the penultimate point of $u''$ is
$V = \bx(c) + \sum_{j=1}^{i} (\alpha_j'' + h_j  \cdot\epsilon ) a_j(c) \ge 0$.
Let $a_k$ be the transition appearing at the end of $u''$ for some $k \le i$. Note that
at this point $a_k$ is fired with fraction $\epsilon$ and so by definition $h_k < \extra_k$.
Now since $c$ becomes negative after firing $a_k$ with fraction $\epsilon$, it must be the case that $a_k$ decrements $c$ and $V + \epsilon a_k(c) < 0$.

Since $k \le i$, by Fact 1, it follows that $\bx_{i}(c) = \bx(c) + \sum_{j=1}^{i} \alpha_j a_j(c) > 0$. 
Hence, we then have
\begin{align*}
	&V + \epsilon a_k(c) \\
    &\quad= \bx(c) + \sum_{j=1}^{i} (\alpha_j'' + h_j \cdot \epsilon) a_j(c) + \epsilon a_k(c) & \text{ Substitute value for } V\\
	&\quad= \bx_i(c) - \sum_{j=1}^i \alpha_j a_j(c) + \sum_{j=1}^{j} (\alpha_j'' + h_j \cdot \epsilon) a_j(c) + \epsilon a_k(c) & \text{ Substitute value for } \bx(c)\\
	&\quad= \bx_i(c) -\left(\sum_{j=1}^{i} (\extra_j - h_j) \cdot \epsilon a_j(c) - \epsilon a_k(c) \right) & \text { Substitute value for } \alpha_j''\\
	&\quad\ge \bx_i(c) - \left (\sum_{j=1}^i \extra_j \right) \cdot \epsilon (\max\{|a_1(c)|, |a_2(c)|, \ldots, |a_i(c)|\}) 
\end{align*}

Now since $\bx_i(c) > 0$, by definition of $\epsilon$, it follows that the last 
value is $> 0$ and so we have that $V + \epsilon a_k(c) > 0$, which leads to a contradiction.
This completes the proof of construction of $\alpha' u'$.

\subsubsection*{Constructing $\ \gamma'v'$}

We now describe how we can redistribute the firing fractions in $\gamma v$ to accommodate for the 
additional occurrences of each transition in $v'$. The proof of this part is the same as the proof of the previous part, except we apply everything in reverse. 

Let $A = \{a_1,\dots,a_n\}$. Note that $v'$ is some word which is contained in 
$b_1 (A \setminus \{b_1\})^* b_2 (A \setminus \{b_1,b_2\})^* \dots b_{n-1}(A \setminus \{b_1,\dots,b_{n-1}\})^* b_n$.
Now, we consider the transitions $b_1^\dagger, b_2^\dagger, \ldots b_n^\dagger$
where each $b_i^\dagger = - b_i$. Let $v'^\dagger$ be the word obtained by first reversing $v'$
and then replacing each $b_i$ with $b_i^\dagger$. Note that
$v'^\dagger$ is some word which is contained in $b_n^\dagger \{b_n^\dagger\}^* b_{n-1}^\dagger 
\{b_{n}^\dagger,b_{n-1}^\dagger\}^* \dots b_2^\dagger \{b_n^\dagger,\dots,b_2^\dagger\}^* b_1^\dagger$.


Let us denote the run $r_3 = \by' \act{\gamma v} \by$ by
$\by' = \by_0 \xrightarrow{\gamma_1 b_1} \by_1 \xrightarrow{\gamma_2 b_2} \by_2 \ldots \by_{n-1} \xrightarrow{\gamma_n b_n} \by_n = \by$.
By assumption on $r_2$, we are guaranteed that for any $1 \le i \le n$ and any counter $c$, if $\bx_i(c) = 0$ 
then $\bx_j(c) = 0$ for all $j \ge i$. Now, by definition of a run, it follows that 
\[\by_n = \by \xrightarrow{\gamma_n b_n^\dagger} \by_{n-1} \xrightarrow{\gamma_{n-1} b_{n-1}^\dagger} \by_{n-2} \dots \by_1 \xrightarrow{\gamma_1 b_1^\dagger} \by_0 = \by' \]
is also a run such that if at some configuration along this new run
the value in some counter is \emph{non-zero}, then the value in that counter stays non-zero
till the end of this new run. 

This means that by the previous construction, we now have a sequence of fractions $\gamma'^\dagger$
such that $\by \xrightarrow{\gamma'^\dagger v'^\dagger} \by'$.
If we let $\gamma'$ be the sequence obtained by reversing $\gamma'^\dagger$,
it follows that $\by' \xrightarrow{\gamma' v'} \by$.
This completes the proof of construction of $\gamma' v'$.

\subsubsection*{Constructing $\ \beta' \centr(w')$} 

We now describe how we can redistribute the firing fractions in $\beta \centr(w)$ to accommodate for the 
additional occurrences of each transition in $\centr(w')$. 

Recall that the set of transitions appearing in $\centr(w)$ is exactly $A = \{a_1,\dots,a_n\}$. Since $\centr(w) \preceq \centr(w')$ and since $\centr(w) \in A^*$, it follows that the set of transitions appearing in $\centr(w')$  is also exactly $A$. 

Let us denote the run $r_2 = \bx' \act{\beta \centr(w)} \by'$ by
$\bx' = \bx_0 \xrightarrow{\beta_1 d_1} \bx_1 \xrightarrow{\beta_2 d_2} \bx_2 \dots \bx_{m-1} \xrightarrow{\beta_m d_m} \bx_m = \by'$. By assumption on $r_2$,
we are guaranteed the following fact 
\begin{quote}
	\textbf{Fact 2: } For any counter $c$, if $\bx'(c) > 0$ or $\by'(c) > 0$  then $\bx_j(c) > 0$ for all $j$. 	
\end{quote}

Now, note that since $\centr(w) \preceq \centr(w')$, it follows that $\centr(w)$ is obtained from $\centr(w')$ by deleting some letters. This is the same as saying that $\centr(w')$ is obtained from
$\centr(w)$ by inserting some letters.
With this observation in mind, we now construct the desired firing sequence $\beta' \centr(w')$ by a case distinction on $\centr(w')$.

\paragraph{Inserting one letter: } First, we consider the case when $\centr(w')$ is obtained by just inserting a single letter
into $\centr(w)$, i.e., $\centr(w') = d_1 d_2 \ldots d_i d d_{i+1} \ldots d_m$ where $\centr(w) = d_1 d_2 \ldots d_m$. Since the set of letters in $\centr(w)$ and $\centr(w')$ are the same,
it follows that $d$ appears either in $d_1 \ldots d_i$ or in $d_{i+1} \ldots d_m$.
Let $d_j$ be either the last appearance of $d$ in $d_1 \ldots d_i$
or the first appearance of $d$ in $d_{i+1} \ldots d_m$.
We now show that it is possible to find a sequence of fractions $\beta'$
such that $\bx' \xrightarrow{\beta' \centr(w')} \by'$ and the following property is satisfied: 
If the value of some counter is non-zero either at $\bx'$ or $\by'$, then it stays non-zero along all the configurations along this run.

To this end, we let $\epsilon$ be a non-zero fraction that is smaller than all the fractions $\beta_1, \dots, \beta_m$ and also all the numbers $\bx_{\ell}(c)$ for any counter $c$ and any $1 \le \ell \le m$ satisfying $\bx_\ell(c) > 0$. Furthermore, we also require that $\epsilon \cdot |a_\ell(c)| < \bx_{\ell'}(c)$, for any $1 \le \ell \le n$ and any $1 \le \ell' \le m$ and any counter $c$ such that $\bx_{\ell'}(c) > 0$.
Note that such an $\epsilon$ can always be chosen.
We then let $\beta' = \beta_1', \beta_2', \dots, \beta_i', \beta_{sp}', \beta_{i+1}', \dots, \beta_m'$ as follows: $\beta'_\ell = \beta_\ell$ for every $\ell \neq j$, $\beta'_j = \beta_j - \epsilon$
and $\beta'_{sp} = \epsilon$. 

Notice that the net effect of firing $\centr(w')$ with the sequence of fractions $\beta'$
is the same as firing $\centr(w)$ with the sequence of fractions $\beta$. So the only thing left to prove is that while firing the transitions in the manner specified above no counter goes below zero.

Let us denote the run obtained by firing $\centr(w')$ with the sequence of fractions $\beta'$
from $\bx'$ as 
\[\bx' = \bx'_0 \xrightarrow{\beta_1' d_1} \bx'_1 \xrightarrow{\beta_2' d_2} \bx'_2 \dots 
\bx'_i \xrightarrow{\beta'_{sp} d} \bx'_{sp} \xrightarrow{\beta'_{i+1} d_{i+1}} \bx'_{i+1} \dots
\bx'_{m-1} \xrightarrow{\beta'_m d_m} \bx'_m = \by' \]

We need to prove that for any counter $c$, $\bx'_k(c)$ is never negative for any $k \in \{0,1,\dots,m\} \cup \{sp\}$
as well as that if $\bx'(c) > 0$ or if $\by'(c) > 0$ then $\bx'_k(c) > 0$ for any $k \in \{0,1,\dots,m\} \cup \{sp\}$. Let us prove this by induction on $k$. Clearly this is true for the base case of $k = 0$.
Suppose we have proved it for all numbers till some $k-1$ and we want to prove it for some $k$.

\textbf{Case 1: } Suppose $j$ comes before $sp$. Then, 

\begin{itemize}
	\item Suppose $k$ comes before $j$. 
	Then, notice that $\beta'_\ell = \beta_\ell$ for all $\ell < k$.
	Hence, $\bx'_k = \bx_k$ and so the claim follows by Fact 2.
	\item Suppose $k = j$. Then, $\beta'_j = \beta_j - \epsilon$.
	By the previous case it follows that $\bx'_{j-1} = \bx_{j-1}$ 
	and so $\bx'_j = \bx_j - \epsilon d_j$. By construction, $\epsilon$ was picked
	so that $\bx_j(c) - \epsilon d_j(c) > 0$ for any $c$ such that $\bx_j(c) > 0$.
	It follows that the claim is true for the configuration $\bx'_j$.   
	\item Suppose $k$ comes between $j$ and $sp$. Then, $\beta'_k = \beta_k$.
	It follows by the above two cases that $\bx'_{k-1} = \bx_{k-1} - \epsilon d_j$
	and so $\bx'_k = \bx_{k} - \epsilon d_j$. By construction, $\epsilon$ was picked
	so that $\bx_k(c) - \epsilon d_j(c) > 0$ for any $c$ such that $\bx_k(c) > 0$.
	It follows that the claim is true for the configuration $\bx'_k$.
	\item Suppose $k = sp$. Then, $\beta'_{sp} = \epsilon$. 
	The position before $sp$ is $i$.
	By the previous case we have that $\bx'_i = \bx_i - \epsilon d_j$.
	Since $d_j = d$ and $d$ is the transition that appears at position $sp$,
	it follows that $\bx'_{sp} = \bx'_i + \epsilon d_j = \bx_i$. 
	The claim then follows by Fact 2.
	\item Suppose $k$ comes after $sp$. Then, notice that $\beta'_k = \beta_k$.
	By the previous case we have that $\bx'_{k-1} = \bx_{k-1}$ and so $\bx'_k = \bx_k$.
	Hence the claim follows by Fact 2.
\end{itemize}

\textbf{Case 2: } Suppose $sp$ comes before $j$. The proof here is similar to the case above, except the $- \epsilon d_j$ trailing term in each of the configurations will be replaced by $\epsilon d_j$.

This finishes the proof for the case in which $\centr(w')$ is obtained from $\centr(w)$ by
inserting one letter.

\paragraph{Inserting multiple letters: } Suppose $\centr(w')$ is obtained from $\centr(w)$
by inserting multiple letters, i.e., $\centr(w') = U_1 e_1 U_2 e_2 \dots U_m e_m U_{m+1}$
where $\centr(w) = U_1 U_2 \dots U_m U_{m+1}$. Then, consider the sequence
of words
$w_0 = \centr(w)$ and 
\[w_{i+1} = U_1 e_1 U_2 e_2 \dots U_{i+1} e_{i+1} U_{i+2} U_{i+3} \dots U_m U_{m+1}\]
Note that each $w_{i+1}$ is obtained from $w_i$ by inserting one letter.
Starting from $w_0$, by repeatedly applying the previous case, it follows that for each $w_{i+1}$, we can find 
a sequence of fractions $\beta'_{i+1}$ such that $\bx' \xrightarrow{\beta'_{i+1} w_{i+1}} \by'$
with the property that if at $\bx'$ or $\by'$ the value of a counter is non-zero, then
it stays non-zero along all the configurations of the run. Since $w_m = \centr(w')$,
we have constructed the required sequence of fractions $\beta' = \beta'_m$.

This completes the proof of construction of $\beta' \centr(w')$ and hence also the proof of~\cref{thm:gathering-lifting-runs}.

\section{Completion of Proof of Theorem~\ref{thm:path-scheme-complement-gathering-decomp}}\label[appsec]{app:gathering-decomp-2}
In this section, we complete the proof of~\cref{thm:path-scheme-complement-gathering-decomp} by proving that $\bigcup_{i \in [1,m]} L(\rho_i) = L(X) \setminus (w)^\uparrow_X$.
To this end, for each $i$, by the definition of $\fl{X}$ and $\sigma_i$, it follows that $\rho_i$ is precisely the path-scheme given by
\[a_1 X_{\{a_1\}} a_2 X_{\{a_1,a_2\}} \ldots a_n X_{\overline{c_1}} c_1 X_{\overline{c_2}} c_2 \ldots X_{\overline{c_i}} b_1 X_{A\setminus\{b_1\}}
b_2 X_{A\setminus\{b_1, b_2\}}
\ldots
X_{A\setminus\{b_1, \ldots, b_{n-1}\}}
b_n
\]

It follows that any word $w'$ belongs to $L(\rho_i)$ iff $w'$
can be written as
\[w' = a_1 w_1' a_2 w_2' \dots a_n w'' b_1 w'_{n} \dots w'_{2n-2} b_n\]
where
\begin{itemize}
	\item For each $i \in [1,n-1], w_i' \in L(X_{\{a_1,\ldots,a_i\}}) = \{a_1,\dots,a_i\}^*$,
	\item $w'' \in L(X_{\overline{c_1}} c_1 X_{\overline{c_2}} c_2 \ldots X_{\overline{c_i}}) = L(\sigma_i) \subseteq L(X_A)$,
	\item And for each $i \in [n,2n-2]$, $w_i' \in L(X_{A \setminus \{b_1,\dots,b_{i-n+1}\}}) = (A \setminus \{b_1,\dots,b_{i-n+1}\})^*$.
\end{itemize}

Hence, this means that $w' \in L(\rho_i)$ iff $w' \in L(X)$ and $\centr(w') \in L(\sigma_i)$.
Since $\bigcup_{i \in [1,m]} L(\sigma_i)$ precisely captures the set of all words that are 
not larger than or equal to $\centr(w)$ with respect
to the subword ordering, it follows that $w' \in \bigcup_{i \in [1,m]} L(\rho_i)$ iff $w' \in L(X)$ and $\centr(w) \not \preceq \centr(w')$. The latter is precisely the definition of $L(X) \setminus (w)^\uparrow_X$ and so the proof is complete.

\section{Comments on the lower bound}\label[appsec]{app:lower-bound}
Here is an exponential run from $\bx_n$ to $\by_n$ that is of the form  in \eqref{lower-bound-runs}: Take the run
\[ w_1^{2n}r_1w_2^{4n}r_2\cdots w_h^{2^hn} r_n u\,, \]
hence we choose $\ell_i=2^in$ for $i=1,\ldots,h$. This is clearly not a short run (by our definition of short run), but it is of at most exponential length.

\myparagraph{Intuition} Before we explain in detail how this can be done, let us give some intuition.
The reason our lower bound construction enforces short runs to be maxed out is that during the firing of $w_i$, the only way to put $4$ into $\step_i$ is to always use the maximal possible firable fractions: $t_{i,1}$ needs to empty $x_1$ completely, and then $t_{i,4}$ needs to fill up $y_1$ completely. This is because in an $(i,k)$-short run, we need the full $\tfrac{1}{k}$ fraction during $t_{i,2}t_{i,3}$ and during $t_{i,5}t_{i,6}$ in order to get $\step_i$ up to $4$ in only $k$ iterations of $w_i$.

However, if we allow ourselves to iterate $w_i$ a whole $2k$ number of times
(starting from a configuration that is consistent with $(i,k)$-short runs,
i.e.\ $\langle \tfrac{1}{k};\tfrac{1}{k};0\rangle$) and still make sure that
$t_{i,5}$ and $t_{i,6}$ can fire with fraction at least $\tfrac{1}{k}$, then it
is not necessary for $t_{i,2}$ and $t_{i,3}$ to fire will the full
$\tfrac{1}{k}$ fraction: The $2k$ occurrences of $t_{i,5}t_{i,6}$ can compensate
for whatever $t_{i,2}t_{i,3}$ were not able to add to $\step_i$. Therefore, we can fire $t_{i,1}$ and $t_{i,4}$ so as to only subtract a very small value (we will choose $\tfrac{1}{2k\cdot k}$) from $x_i$. This allows us to fire all copies of $w_i$ and instead of turning $x_i$ from $\tfrac{1}{k}$ into $\tfrac{1}{k\cdot 2^k}$ (as we must in a $(i,k)$-short run), we turn $x_i$ into $\tfrac{1}{2}\cdot\tfrac{1}{k}$. This way, each factor $w_i^{\ell_i}=w_i^{2^in}$ only halves the values in the high-precision counters, instead of applying an exponential function to their denominator.

\myparagraph{Detailed run description}
To show that the above sequence can be fired from $\bx_n$ to $\by_n$, we prove the following claim:
\begin{lemma}\label{exponential-solution:intermediate-steps} For every $k\ge 0$, we can fire
\begin{equation}\left\langle \frac{1}{k}; \frac{1}{k}; 0\right\rangle \xrightarrow{w_i^{2k}} \left\langle \frac{1}{2k}; \frac{1}{k}; 4\right\rangle \label{exponential-solution:intermediate-goal}\end{equation}
with appropriate firing fractions.
\end{lemma}
\begin{proof}
We construct the run one copy of $w_i$ at a time. Specifically, we show by induction on $j=1,\ldots,2k$ that we can fire
\begin{equation} \left\langle \frac{1}{k}; \frac{1}{k}; 0\right\rangle \xrightarrow{w_i^j} \left\langle \frac{1}{k}-\frac{j}{4k\cdot k}; \frac{1}{k}; j\cdot 2\gamma\left(\frac{1}{2k\cdot k}+\frac{1}{k}\right)\right\rangle \label{exponential-solution:induction-hypothesis}\end{equation}
for every rational choice of $\gamma\in[0,1)$. Here, we use the notation for configurations as described in \eqref{lower-bound:notation-configurations}.

Indeed, we can fire the $j$-th copy of $w_i$ by using the following fractions:
\begin{enumerate}
\item $t_{i,1}$ with fraction $\tfrac{1}{2}\cdot \tfrac{1}{2k\cdot k}$.
\item both $t_{i,2}$ and $t_{i,3}$ with fraction $\gamma\cdot \tfrac{1}{2k\cdot k}$
\item $t_{i,4}$ with $\tfrac{1}{2}\cdot \tfrac{1}{2k\cdot k}$
\item both $t_{i,5}$ and $t_{i,6}$ with fraction $\gamma\cdot\tfrac{1}{k}$
\end{enumerate}
This is because then $t_{i,1}$ and $t_{i,4}$ together first subtract $\tfrac{1}{2k\cdot k}$ from $x_i$ and then add $\tfrac{1}{2}\cdot\tfrac{1}{2k\cdot k}$ to $x_i$, meaning overall we subtract $\tfrac{1}{2}\cdot\tfrac{1}{2k\cdot k}=\tfrac{1}{4k\cdot k}$. The effect on $y_i$ is to first subtract $\tfrac{1}{2}\cdot \tfrac{1}{2k\cdot k}$ and then add the same value, hence $y_i$ returns to $\tfrac{1}{k}$. Moreover, the effect on $\step_i$ is that we add
\[ 2\cdot\gamma\cdot\frac{1}{2k\cdot k}+2\cdot\gamma\cdot\frac{1}{k}=2\gamma\left(\frac{1}{2k\cdot k}+\frac{1}{k}\right) \]
Thus, we have established \eqref{exponential-solution:induction-hypothesis}.

Applying \eqref{exponential-solution:induction-hypothesis} to $j=2k$ yields:
\[ \left\langle \frac{1}{k}; \frac{1}{k}; 0\right\rangle \xrightarrow{w_i^{2k}} \left\langle \frac{1}{2k}; \frac{1}{k}; 2k\cdot 2\gamma\left(\frac{1}{2k\cdot k}+\frac{1}{k}\right)\right\rangle=\left\langle \frac{1}{2k}; \frac{1}{k}; \gamma\cdot 4(\tfrac{1}{2k}+1)\right\rangle. \]
Therefore, for an appropriate rational choice of $\gamma\in(0,1]$, we can achieve $\gamma(\tfrac{1}{2k}+4)=4$ and thus \cref{exponential-solution:intermediate-goal}.
\end{proof}

Now that \cref{exponential-solution:intermediate-steps} is established, we can see that for each $i$, applying $w_i^{2^in}$ can be used to put $4$ into $\step_i$ while (i)~halving the values in all precision counters $x_i$ and $x_{i+1},\ldots,x_h$ and $y_{i+1},\ldots,y_h$ and (ii)~not changing $y_i$. With $r_i$, we can then reset the complement counters to zero. Finally, using $u$, we can reset all high-precision counters.

\label{afterbibliography}
\newoutputstream{pagestotal}
\openoutputfile{main.pagestotal.ctr}{pagestotal}
\addtostream{pagestotal}{\getpagerefnumber{afterbibliography}}
\closeoutputstream{pagestotal}


\newoutputstream{todos}
\openoutputfile{main.todos.ctr}{todos}
\addtostream{todos}{\arabic{@todonotes@numberoftodonotes}}
\closeoutputstream{todos}
\end{document}